%% file: main-Fourier-testing.tex
\documentclass[11pt]{article}

\def\colorful{0}
\input{preamble}

\title{Fourier-Based Testing for Families of Distributions}

\author{
Cl\'ement L. Canonne\thanks{Supported by NSF grants CCF-1115703 and NSF CCF-1319788.}\\
Columbia University\\
{\tt ccanonne.d@cs.columbia.edu}\\
\and
Ilias Diakonikolas\thanks{Supported by NSF Award CCF-1652862 (CAREER) and a Sloan Research Fellowship.}\\
University of Southern California\\
{\tt diakonik@usc.edu}\\
\and
Alistair Stewart\thanks{Supported by a USC startup grant.}\\
University of Southern California\\
{\tt stewart.al@gmail.com}
}

\begin{document}

\maketitle

\thispagestyle{empty}
\begin{abstract}
  \input{abstract}
\end{abstract}
\setcounter{page}{0}
\newpage

\section{Introduction}  \label{sec:intro}

\subsection{Background and Motivation} \label{ssec:background}

The prototypical inference question in the area of {\em distribution property testing}~\cite{BFR+:00}
is the following: Given a set of samples from a collection of probability distributions, can we
determine whether these distributions satisfy a certain property?
During the past two decades, this broad
question -- whose roots lie in statistical hypothesis testing~\cite{NeymanP, lehmann2005testing} --
has received considerable attention by the computer science community,
see~\cite{Rub12, Canonne15} for two recent surveys.
After two decades of study, for many properties of interest there exist
sample-optimal testers (matched by information-theoretic lower bounds)
~\cite{Paninski:08, CDVV14, VV14, DKN:15, ADK15, DK16}.

In this work, we focus on the problem of testing whether an unknown distribution
belongs to a given family of discrete {\em structured} distributions.
Let $\mathcal{P}$ be a family of discrete distributions over a total order (e.g., $[n]$)
or a partial order (e.g., $[n]^k$). 
The problem of {\em membership testing for $\mathcal{P}$} is the following:
Given sample access to an unknown distribution $\p$ (effectively supported 
on the same domain as $\mathcal{P}$),
we want to distinguish between the case that $\p \in \mathcal{P}$ versus $\dtv(\p, \mathcal{P}) > \eps$.
(Here,  $\dtv$ denotes the total variation distance between distributions.)
The sample complexity of this problem depends on the underlying family
$\mathcal{P}$. For example, if $\cal P$ contains a single distribution over a domain of size $n$,
the sample complexity of the testing problem is $\Theta(n^{1/2}/\eps^2)$~\cite{CDVV14, VV14, DKN:15, ADK15}.

In this work, we give a general technique to test membership in various distribution families over discrete domains, i.e., to solve the following task:
\begin{framed}
\noindent $\mathfrak{T}(\mathcal{P},\eps)$: given a family of discrete distributions $\mathcal{P}$ 
over some partially or totally ordered set, parameter $\eps \in(0,1]$, and sample access to an
unknown distribution $\p$ over the same domain, how many samples are required 
to distinguish, with probability $3/5$, between the case that $\p \in \mathcal{P}$ versus $\dtv(\p, \mathcal{P}) > \eps$?
\end{framed}
\noindent Before we state our results in full generality, we present concrete applications to 
a number of well-studied distribution families.

\subsection{Our Results} \label{ssec:results}
Our first result is a nearly sample-optimal testing algorithm for sums of independent integer random variables (SIIRVs). 
Formally, an $(n, k)$-SIIRV is a sum of $n$ independent integer random variables each supported
in $\{0, 1, \ldots, k-1\}$. We will denote the set of $(n, k)$-SIIRVs by $\classksiirv[n]{k}$.
SIIRVs comprise a rich class of distributions that arise in many settings. The special case of $k=2$
was first considered by Poisson \cite{Poisson:37} as a non-trivial extension of the Binomial distribution,
and is known as Poisson binomial distribution (PBD). In application domains, SIIRVs have many uses in research areas
such as survey sampling, case-control studies, and survival analysis, see e.g., \cite{ChenLiu:97} for a survey of the many practical uses of these distributions.
In addition to their practical applications, SIIRVs are of fundamental probabilistic interest and have been extensively 
studied in the theory of probability and statistics~\cite{Chernoff:52,Hoeffding:63,DP09, Presman:83,Kruopis:86,BHJ:92, CL10,CGS11}.
We prove:

\begin{theorem}[Testing SIIRVs]\label{theo:testing:ksiirv}
    Given parameters $k,n\in\mathbb{N}$ and sample access to a distribution over $\mathbb{N}$, 
    there exists an algorithm (\cref{algo:ft:effective:support}) for $\mathfrak{T}(\classksiirv[n]{k},\eps)$ which takes 
    \[
        \bigO{\frac{k n^{1/4}}{\eps^2}\log^{1/4}\frac{1}{\eps} + \frac{k^2}{\eps^2} \log^2\frac{k}{\eps}}
    \] samples, and runs in time $n(k/\eps)^{\bigO{k\log(k/\eps)}}$.
\end{theorem}
Prior to our work, no non-trivial\footnote{By the term ``non-trivial'' here we refer to a testing algorithm that uses fewer samples
than just learning the unknown distribution and then checking whether it is close to a distribution in the family.} 
tester was known for $(n, k)$-SIIRVs for any $k>2$. 
\cite{CDGR:16} showed a sample lower bound of $\bigOmega{{k^{1/2}n^{1/4}}/{\eps^2}}$, but their techniques
did not yield any non-trivial sample upper bound. 

For the special case of PBDs ($k=2$), 
Acharya and Daskalakis~\cite{AD15} gave a tester with sample complexity 
$\bigO{\frac{n^{1/4}}{\eps^2}\sqrt{\log1/\eps}+\frac{\log^{5/2}1/\eps}{\eps^6}}$, 
running time $\bigO{\frac{n^{1/4}}{\eps^2}\sqrt{\log1/\eps}+(1/\eps)^{O(\log^2 1/\eps)}}$, and also showed a 
sample lower bound of $\Omega(n^{1/4}/\eps^2)$. The special case of our Theorem~\ref{theo:testing:ksiirv}
for $k=2$ yields an improvement over~\cite{AD15} in both sample size and runtime:

\begin{theorem}[Testing PBDs]\label{theo:testing:[bd}
    Given parameter $n\in\mathbb{N}$ and sample access to a distribution over $\mathbb{N}$, there exists an algorithm (\cref{algo:ft:effective:support}) for $\mathfrak{T}(\mathcal{PBD}_{n},\eps)$ which takes 
    \[
        \bigO{\frac{n^{1/4}}{\eps^2}\log^{1/4}\frac{1}{\eps} + \frac{\log^2 1/\eps}{\eps^2}}
    \] samples, and runs in time $n^{1/4}\cdot\tildeO{{1}/{\eps^2}}+(1/\eps)^{\bigO{\log\log(1/\eps)}}$.
\end{theorem}

Note that the sample complexity of our algorithm is $n^{1/4} \cdot \tilde{O}(1/\eps^2)$, matching the information-theoretic lower bound
up to a logarithmic factor in $1/\eps$. In particular, our algorithm does not incur the extraneous $\Omega(1/\eps^6)$ term
of~\cite{AD15}. Moreover, our runtime has a $(1/\eps)^{\bigO{\log\log(1/\eps)}}$ dependence, as opposed to 
$(1/\eps)^{O(\log^2 1/\eps)}$.
The improved running time relies on a more efficient computational ``projection step''  in our general framework,
which leverages the geometric structure of Poisson Binomial distributions. 

We remark that the guarantees provided by the above two theorems 
are actually stronger than the usual property testing one. Namely, whenever the algorithm returns \accept, 
then it also provides a (proper) hypothesis $\h$ such that $\dtv(\p,\h)\leq \eps$ with probability at least $3/5$.

A broad generalization of PBDs to the high-dimensional setting is the family of
Poisson Multinomial Distributions (PMDs).
Formally, an $(n, k)$-PMD is any random variable of the form $X = \sum_{i=1}^n X_i$,
where the $X_i$'s are independent random vectors supported on the set
$\{e_1, e_2, \ldots, e_k \}$ of standard basis vectors in $\R^k$.
We will denote by $\classpmd[n]{k}$ the set of $(n, k)$-PMDs.
PMDs comprise a broad class of discrete distributions of fundamental importance in computer science, probability, and statistics.
A large body of work in the probability and statistics literature has been devoted to the study of the behavior
of PMDs under various structural conditions~\cite{Barbour88, Loh92, BHJ:92, Bentkus:03, Roos99, Roos10}.
PMDs generalize the familiar multinomial distribution, and describe many distributions
commonly encountered in computer science (see, e.g.,~\cite{DaskalakisP07, DaskalakisP08, Valiant08stoc, ValiantValiant:11}).
Recent years have witnessed a flurry of research activity on PMDs and related distributions,
from several perspectives of theoretical computer science,
including learning~\cite{DDS12stoc, DDOST13focs, DKS:16, DKT15, DKS15b},
property testing~\cite{Valiant08stoc, VV10b, ValiantValiant:11},
 computational game theory~\cite{DaskalakisP07, DaskalakisP08, BorgsCIKMP08, DaskalakisP09, DaskalakisP2014, GT14,CDS:17},
 and derandomization~\cite{GMRZ11, BDS12, De15, GKM15}. We prove the following:
 
 \begin{theorem}[Testing PMDs]\label{theo:testing:pmd}
    Given parameters $k,n\in\mathbb{N}$ and sample access to a distribution over $\mathbb{N}^k$, there exists an algorithm (\cref{algo:pmd:tester}) for $\mathfrak{T}(\classpmd[n]{k},\eps)$ which takes 
    \[
        \bigO{\frac{n^{(k-1)/4} k^{2k}}{\eps^2}\log(k/\eps)^k}
    \] samples, and runs in time $n^{O(k^3)} \cdot (1/\eps)^{O(k^3\frac{\log(k/\eps)}{\log\log(k/\eps)})^{k-1}}$ or alternatively in time $n^{O(k)} \cdot  2^{O(k^{5k} \log(1/\eps)^{k+2})}$.
\end{theorem}

For the sake of intuition, we note that Theorem~\ref{theo:testing:pmd} is particularly interesting in the regime that $n$ 
is large and $k$ is small. Indeed, the sample complexity of testing PMDs is inherently {\em exponential} in $k$:
We prove a sample lower bound of $\Omega_k( n^{{(k-1)}/{4}}/\eps^2 )$ (\cref{theo:lb:pmd}),\footnote{Here, we use the notation $\Omega_k(\cdot)$, $O_k(\cdot)$ to indicate that the parameter $k$ is seen as a constant.}{} nearly-matching our upper bound for constant $k$.

Finally, we demonstrate the versatility of our techniques by obtaining (\cref{sec:log:concaves}) a testing algorithm for discrete log-concavity.
Log-concave distributions constitute a broad and flexible non-parametric family
that is extensively used in modeling and inference~\cite{Walther09}.
In the discrete setting, log-concave distributions encompass a range of fundamental
types of discrete distributions, including binomial, negative binomial,
geometric, hypergeometric, Poisson, Poisson Binomial, hyper-Poisson,
P\'{o}lya-Eggenberger, and Skellam distributions. 
Log-concave distributions have been studied in a wide range of different contexts including
economics \cite{An:95}, statistics and probability theory (see~\cite{SW14-survey} for a recent survey),
theoretical computer science~\cite{LV07}, and algebra, combinatorics and geometry~\cite{Stanley:89}.
We will denote by $\classlogconcave_n$ the class of log-concave distributions over $[n]$.
We prove:

\begin{theorem}[Testing Log-Concavity]\label{theo:testing:lcv}
Given a parameter $n\in\mathbb{N}$ and sample access to a distribution over $\mathbb{N}$, 
there exists an algorithm (\cref{algo:log-con:tester}) for $\mathfrak{T}(\classlogconcave_n,\eps)$ which takes 
\[
\bigO{\frac{\sqrt{n}}{\eps^2}} + \tildeO{\frac{1}{\eps^{5/2}}}
\] samples, and runs in time $O(\sqrt{n} \cdot \poly(1/\eps))$.
\end{theorem}

Our discrete log-concavity tester improves on previous work in terms of both sample and time complexity.
Specifically, \cite{ADK15} gave a log-concavity tester with sample complexity $\bigO{{\sqrt{n}}/{\eps^{2}}+{1}/{\eps^5}}$,
while~\cite{CDGR:16} obtained a tester with sample complexity $\tildeO{{\sqrt{n}}/{\eps^{7/2}}}$. Our sample complexity dominates
both these bounds, and is significantly better when $\eps$ is small.
The algorithms in~\cite{ADK15, CDGR:16} run in $\poly(n/\eps)$ time, as they involve 
solving a linear program of $\poly(n/\eps)$ size. In contrast, the running time of our algorithm is 
{\em sublinear} in $n$.

\subsection{Our Techniques and Comparison to Previous Work} \label{ssec:techniques}
\new{All the testing algorithms in this paper follow from a simple and general technique that may be of broader interest.}
The common property of the underlying distribution families $\mathcal{P}$ that allows for our unified testing approach
is the following: Let $\p$ be the probability mass function of any distribution in $\mathcal{P}$. Then, {\em the Fourier transform
of $\p$ is approximately sparse}, in a well-defined sense. 

For concreteness,  we elaborate on our technique for the case of SIIRVs.
The starting point of our approach is the observation from~\cite{DKS:16} that $(n,k)$-SIIRVs -- 
in addition to having a relatively small effective support -- also have an approximately sparse Fourier representation. 
Roughly speaking, most of their Fourier mass is concentrated on a small subset of Fourier coefficients, which can be computed efficiently.

This suggests the following natural approach to testing $(n,k)$-SIIRVs: first, identify the effective support $I$ of the distribution $\p$ 
and check that it is appropriately small. \new{If it is not, then reject.}
Then, compute the corresponding small subset $S$ of the Fourier domain, 
and check that almost no Fourier mass of $\p$ lies outside $S$. Otherwise, one can safely reject, as this is a certificate that $\p$ is not an $(n,k)$-SIIRV. 
Combining the two steps, one can show that learning the Fourier transform of $\p$ (in $L_2$-norm) on this small subset $S$ only, 
is sufficient to learn $\p$ itself in total variation distance. The former goal can be performed with relatively few samples, as $S$ is sufficiently small.

At this point, we have obtained a distribution $\h$ -- succinctly represented by its Fourier transform on $S$ -- 
such that $\p$ and $\h$ are close in total variation distance. 
It only remains to perform a computational ``projection step'' to verify that $\h$ itself is close to some $(n,k)$-SIIRV. 
This will clearly be the case if indeed $\p\in\classksiirv[n]{k}$.

Although the aforementioned approach forms the core of our SIIRV testing algorithm (\cref{algo:ksiirv:tester}), 
the actual tester has to address separately the case where $\p$ has small variance, 
which can be handled by a testing-via-learning approach. 
Our main contribution is thus to describe how to efficiently perform the second step, i.e., the Fourier sparsity testing. 
This is done in~\cref{theo:ft:effective:support}, which describes a simple algorithm to perform this step. The algorithm proceeds 
by essentially considering the Fourier coefficients of the empirical distribution (obtained by taking a small number of samples).
Interestingly, the main idea underlying~\cref{theo:ft:effective:support} 
is to avoid analyzing directly the behavior of these Fourier coefficients -- which would naively require too high a time complexity. 
Instead, we rely on Plancherel's identity and reduce the problem to the analysis of a different task:  
that of the sample complexity of {\em $L_2$ identity testing} (\cref{prop:l2:identity:tester}). By a tight analysis of this $L_2$ tester, 
we get as a byproduct that several Fourier quantities of interest (of our empirical distribution) simultaneously enjoy 
good concentration -- while arguing concentration of each of these terms separately would yield a suboptimal time complexity. 

A nearly identical method works for PMDs as well. Moreover, our approach can be abstracted 
to yield a general testing framework, as we explain in~\cref{sec:general:testing}. It is interesting to
remark that the Fourier transform has been used to learn PMDs and SIIRVs~\cite{DKS:16, DKT15, DKS15b, DDKT16}, 
and therefore it may not be entirely surprising that it has applications to testing as well. 
Notably, our Fourier testing technique gives an improved and nearly-optimal algorithms for log-concavity, for which no Fourier learning algorithm was known. More generally, testing membership to a class using the Fourier transform is significantly more challenging than learning. 
A fundamental difference is that in the testing setting we need to handle distributions that do \emph{not} belong to the class (e.g., SIIRVs, PMDs), 
but are far from the class in an arbitrary way. In contrast, learning algorithms work under the promise 
that the distribution is in the underlying class, and thus can leverage the specific structure. 

\vspace{-0.2cm}

\paragraph*{Testing via the Fourier Transform: the Advantage}
One may wonder how the detour via the Fourier transform enables us to obtain better sample complexity than an approach purely based on $L_2$ testing.  Indeed, all distributions in the classes we consider, crucially, have small $L_2$ norm. 
For testing identity to such a distribution $\p$, the standard $L_2$ identity tester (see, e.g.,~\cite{CDVV14} or~\cref{prop:l2:identity:tester}), 
which works by checking how large the $L_2$-distance between the empirical and the hypothesis distribution is, will be optimal. 
We can thus test membership of a class of such distributions by (i) learning $\p$ assuming it belongs to the class, and then (ii) test whether what we learned is indeed close to $\p$ using the $L_2$ identity tester. The catch is that, in order to get guarantees in $L_1$-distance using this approach, would require us to learn to very small $L_2$ distance (because of the Cauchy--Schwarz inequality). In particular, if the unknown distribution $\p$ has support size $N$, we would have to learn to $L_2$ distance ${\eps}/{\sqrt{N}}$ in (i), and then in (ii) test that we are within $L_2$-distance ${\eps}/{\sqrt{N}}$ of the learned hypothesis.

However, if a distribution $\p$ has a sparse discrete Fourier transform (whose effective support is known), then suffices h to estimate only these few Fourier coefficients~\cite{DKS:16,DKS15c}. This step enables us to learn $\p$ in (i) not just to within $L_1$-distance $\eps$, but indeed (crucially) within $L_2$-distance $\frac{\eps}{\sqrt{N}}$ with good sample complexity. Additionally, the identity testing algorithm can be put into a simpler form for a hypothesis with sparse Fourier transform, as previously mentioned. Now, the tester has higher sample complexity, roughly $\sqrt{N}/\eps^2$; but if it accepts, then we have learned the distribution $\p$ to within $\eps$ total variation distance, with much fewer samples than the $\bigOmega{N/\eps^2}$ required for arbitrary distributions over support size $N$.
\new{Lastly, we note that we can replace the support size $N$ in the above description by the size of the {\em effective support}, i.e., the smallest
set that contains $1-O(\eps)$ fraction of the mass. Doing so for the case of $(n,k)$-SIIRVs leads to a sample complexity 
proportional to $n^{1/4}$, instead of $n^{1/2}$.}

\subsection{Organization} \label{ssec:organization}
\new{The rest of the paper is organized as follows: In~\cref{sec:prelim}, we set up notations and provide definitions as well as standard results relevant to our purposes. \cref{sec:fourier:support:testing} contains the details of one of the main subroutines our testers rely on, namely for \emph{Fourier sparsity testing}. We then give and analyze in~\cref{sec:siirv:testing} our Fourier-based tester for SIIRVs. In~\cref{sec:general:testing}, we abstract and generalize this approach to obtain a general tester applicable to any class of distributions which enjoys good Fourier sparsity. \cref{sec:pmd:testing} then contains our tester for Poisson Multinomial Distributions, which we get by extending our general technique to higher dimensions (this tester is complemented in~\cref{sec:lower:bounds} by our sample complexity lower bound on testing PMDs). Finally, we focus in~\cref{sec:log:concaves} on the class of log-concave distributions, leveraging our Fourier-based tools to obtain a tester for this class.

All omitted proofs can be found in~\cref{app:omitted}. In~\cref{appendix:log:concave}, we analyze the sample complexity of learning discrete log-concave distributions via the Maximum Likelihood Estimator, a result that we use for our log-concavity tester and which may be of independent interest.}


\section{Preliminaries} \label{sec:prelim}
\input{preliminaries}

\section{Testing Effective Fourier Support}\label{sec:fourier:support:testing}
\input{fouriertesting}

\section{The SIIRV Tester}\label{sec:siirv:testing}
\input{siirvtesting}

\section{The General Tester}\label{sec:general:testing}
\input{general}

\section{The PMD Tester}\label{sec:pmd:testing}
\input{pmdtesting}

\section{The Discrete Log-Concavity Tester}\label{sec:log:concaves}
\input{logconcaves}

\section{Lower Bound for PMD Testing}\label{sec:lower:bounds}
\input{lowerbounds}

\clearpage
\bibliographystyle{alpha}
\bibliography{allrefs}

\clearpage
\appendix

\input{app-omitted}
\input{app-logconcavelearning}

\end{document}

%% file: preamble.tex
\usepackage{amsthm,amsfonts,amsmath,amssymb,epsfig,xcolor,float,graphicx,verbatim}
\usepackage{multirow}
\usepackage{times}

\usepackage{framed}

\usepackage{hyperref}
\usepackage[capitalize]{cleveref} \usepackage{algorithmicx,algpseudocode,algorithm}

\usepackage{xspace}
\usepackage{fullpage}

\def\nnewcolor{1}
\ifnum\nnewcolor=1

\fi
\ifnum\nnewcolor=0

\fi

\ifnum\colorful=1
\newcommand{\new}[1]{{\color{red} #1}}

\else
\newcommand{\new}[1]{{#1}}

\fi

\makeatletter
\renewcommand{\section}{\@startsection{section}{1}{0pt}{-12pt}{5pt}{\large\bf}}
\renewcommand{\subsection}{\@startsection{subsection}{2}{0pt}{-12pt}{-5pt}{\normalsize\bf}}
\renewcommand{\subsubsection}{\@startsection{subsubsection}{3}{0pt}{-12pt}{-5pt}{\normalsize\bf}}
\makeatother

\newtheorem{theorem}{Theorem}[section]

\newtheorem{lemma}[theorem]{Lemma}
\newtheorem{proposition}[theorem]{Proposition}

\newtheorem{claim}[theorem]{Claim}
\crefname{claim}{Claim}{Claims}
\crefname{lemma}{Lemma}{Lemmas}
\newtheorem{fact}[theorem]{Fact}

\theoremstyle{definition}
\newtheorem{definition}[theorem]{Definition}

\newcommand{\R}{\mathbb{R}}
\newcommand{\C}{\mathbb{C}}
\newcommand{\Z}{\mathbb{Z}}
\newcommand{\N}{\mathbb{N}}
\newcommand{\E}{\mathbb{E}}

\newcommand{\poly}{\mathrm{poly}}

\newcommand{\p}{\mathbf{P}}
\newcommand{\q}{\mathbf{Q}}
\newcommand{\h}{\mathbf{H}}

\newcommand{\dtv}{d_{\mathrm TV}}

\newcommand{\wh}[1]{{\fourier{#1}}}

\newcommand{\var}{\Var}

\newcommand{\ignore}[1]{}

\newcommand{\eps}{\epsilon}

\newcommand{\abs}[1]{\lvert#1\rvert}

\newcommand{\norm}[1]{\lVert#1\rVert}

\newcommand{\Var}{\mathop{\textnormal{Var}}\nolimits}

\newcommand{\Poi}{\mathop{\textnormal{Poi}}\nolimits}

\newcommand{\eqdef}{\stackrel{{\mathrm {\footnotesize def}}}{=}}

\newcommand{\distribs}[1]{\Delta\!\left(#1\right)} \newcommand{\normone}[1]{{\norm{#1}}_1}
\newcommand{\normtwo}[1]{{\norm{#1}}_2}

\renewcommand{\abs}[1]{\left\lvert #1 \right\rvert}
\newcommand{\dabs}[1]{\lvert #1 \rvert}
\newcommand{\setOfSuchThat}[2]{ \left\{\; #1 \;\colon\; #2\; \right\} } 			\newcommand{\clg}[1]{\left\lceil #1 \right\rceil}
\newcommand{\flr}[1]{\left\lfloor #1 \right\rfloor}
\newcommand{\accept}{\textsf{accept}\xspace}

\newcommand{\reject}{\textsf{reject}\xspace}

\newcommand{\bigO}[1]{{O\left( #1 \right)}}
\newcommand{\bigTheta}[1]{{\Theta\left( #1 \right)}}
\newcommand{\bigOmega}[1]{{\Omega\left( #1 \right)}}
\newcommand{\tildeO}[1]{\tilde{O}\left( #1 \right)}

\providecommand{\poly}{\operatorname*{poly}}

\newcommand{\classpbd}[1][n]{\ensuremath{\mathcal{PBD}_{#1}}\xspace}
\newcommand{\classpmd}[2][n]{\ensuremath{\mathcal{PMD}_{#1,#2}}\xspace}
\newcommand{\classksiirv}[2][n]{\ensuremath{\mathcal{SIIRV}_{#1,#2}}\xspace}
\newcommand{\classlogconcave}{\ensuremath{\mathcal{LCV}}\xspace}
\newcommand{\fourier}[1]{\widehat{#1}}

\newcommand{\bracketing}[3][\operatorname{d}_{\rm H}]{\mathcal{N}_{[ ]}(#2,#3,#1)}
\newcommand{\hellinger}[2]{{\operatorname{d_{\rm{}H}}\left({#1, #2}\right)}}
\newcommand{\lp}[1][1]{\ell_{#1}}
\newcommand{\class}{\ensuremath{\mathcal{C}}\xspace} \newenvironment{proofof}[1]{\begin{proof}[Proof of {#1}]}{\end{proof}}
\newcommand{\totalvardist}[2]{{\dtv\left({#1, #2}\right)}}
\newcommand{\shortexpect}{\mathbb{E}}
\newcommand{\proba}{\Pr}
\newcommand{\probaOf}[1]{\proba\left[\, #1\, \right]}

\newcommand{\littlesum}{\mathop{\textstyle \sum}}

%% file: abstract.tex
We study the general problem of testing whether an unknown distribution belongs to a specified family of distributions.
More specifically, given a distribution family $\mathcal{P}$ and sample access to an unknown discrete distribution $\p$,
we want to distinguish (with high probability) between the case that $\p \in \mathcal{P}$ and the case
that $\p$ is $\eps$-far, in total variation distance, from every distribution in $\mathcal{P}$.
This is the prototypical hypothesis testing problem that has received significant attention in statistics and, 
more recently, in theoretical computer science.

The sample complexity of this general inference task depends on the underlying family  $\mathcal{P}$.
The gold standard in distribution property testing is to design sample-optimal and computationally efficient algorithms for this task. 
The main contribution of this work is a simple and general testing technique that is applicable to all distribution families 
whose \emph{Fourier spectrum} satisfies a certain approximate \emph{sparsity} property. To the best of our knowledge, 
ours is the first use of the Fourier transform in the context of distribution testing.

We apply our Fourier-based framework to obtain near sample-optimal and computationally efficient 
testers for the following fundamental distribution families:
Sums of Independent Integer Random Variables (SIIRVs), Poisson Multinomial Distributions (PMDs), 
and Discrete Log-Concave Distributions. For the first two, ours are the first non-trivial testers in the literature,
vastly generalizing previous work on testing Poisson Binomial Distributions.
For the third, our tester improves on prior work in both sample and time complexity.

%% file: preliminaries.tex
\new{We begin with some standard notations and definitions, as well as basics of Fourier analysis and results from Probability that we shall use throughout the paper. We also state two structural results on SIIRVs, which will be useful to us in~\cref{sec:siirv:testing}.} For $m \in \N$, we write $[m]$ for the set $\{0,1,\dots,m-1\}$, and $\log$ (resp. $\ln$) for the binary logarithm (resp. the natural logarithm).

\paragraph{Distributions and Metrics} 
A probability distribution over (discrete) domain $\Omega$ is a function $\p\colon\Omega\to[0,1]$ such that $\normone{\p}\eqdef \sum_{\omega\in\Omega}\p(\omega)=1$; we denote by $\distribs{\Omega}$ the set of all probability distributions over domain $\Omega$. 
Recall that for two probability distributions $\p,\q\in\distribs{\Omega}$, their \emph{total variation distance} (or statistical distance) is defined as 
$
    \dtv(\p,\q) \eqdef \sup_{S\subseteq\Omega} (\p(S)-\q(S)) = \frac{1}{2}\sum_{\omega\in\Omega} \abs{\p(\omega)-\q(\omega)},
$
i.e. $\dtv(\p,\q) = \frac{1}{2}\normone{\p-\q}$. Given a subset $\mathcal{P}\subseteq \distribs{\Omega}$ of distributions, the \emph{distance from $\p$ to $\mathcal{P}$} is then defined as $\dtv(\p,\mathcal{P})\eqdef \inf_{\q\in\mathcal{P}} \dtv(\p,\q)$. If $\dtv(\p,\mathcal{P}) > \eps$, we say that $\p$ is \emph{$\eps$-far} from $\mathcal{P}$; otherwise, it is \emph{$\eps$-close}.

\paragraph*{Property Testing}
We work in the standard setting of distribution testing: a \emph{testing algorithm for a property $\mathcal{P}\subseteq\distribs{\Omega}$} is an algorithm which, granted access to independent samples from an unknown distribution $\p\in\distribs{\Omega}$ as well as distance parameter $\eps\in(0,1]$, outputs either \accept or \reject, with the following guarantees.
\begin{itemize}
  \item if $\p\in\mathcal{P}$, then it outputs \accept with probability at least $3/5$;
  \item if $\dtv(\p,\mathcal{P})>\eps$, then it outputs \reject with probability at least $3/5$.
\end{itemize}
The two measures of interest here are the \emph{sample complexity} of the algorithm (i.e., the number of samples from the distribution it takes in the worst case), and its running time.

\paragraph*{Classes (Properties) of Distributions} We now recall the definition of the three classes of discrete distributions central to this work, which all extend the family of Binomial distributions: the first two, by allowing each summand to be non-identically distributed:
\begin{definition}\label{def:siirv}
Fix any $k\geq 2$. We say a random variable $X$ is a \emph{$(n,k)$-Sum of Independent Integer Random Variables ($(n,k)$-SIIRV)} with parameter $n\in\N$ if it can be written as $X=\sum_{j=1}^n X_j$, where $X_1\dots,X_n$ are independent, non-necessarily identically distributed random variables taking value in $[k]=\{0,1,\dots,k-1\}$. We denote by $\classksiirv[n]{k}$ the class of all such $(n,k)$-SIIRVs.
\end{definition}
\noindent (The class of \emph{Poisson Binomial Distributions}, denoted $\classpbd[n]$, corresponds to the case $k=2$, that is $2$-SIIRVS. Equivalently, this is the generalization of Binomials where each Bernoulli summand is allowed to have its own parameter). A different type of generalization is that of Poisson Multinomial Distributions, where each summand is a random variable supported on the $k$ vectors of the standard basis of $\R^k$, instead of $[k]$:
\begin{definition}\label{def:pmd}
Fix any $k\geq 2$. We say a random variable $X$ is a \emph{$(n,k)$-Poisson Multinomial Distribution ($(n,k)$-PMD)} with parameter $n\in\N$ if it can be written as $X=\sum_{j=1}^n X_j$, where $X_1\dots,X_n$ are independent, non-necessarily identically distributed random variables taking value in $\{e_1,\dots,e_k\}$ (where $(e_i)_{i\in[k]}$ is the canonical basis of $\R^k$). We denote by $\classpmd[n]{k}$ the class of all such $(n,k)$-PMDs.
\end{definition}

\noindent Lastly, we recall the definition of discrete log-concavity.
\begin{definition}\label{def:logconcave}
  A distribution $\p$ over $\Z$ is said to be \emph{log-concave} if it satisfies the following conditions: \textsf{(i)} for any $i < j < k$ such that $\p(i)\p(k) > 0$, $\p(j) > 0$; and \textsf{(ii)} for all $k\in\Z$, $\p(k)^2 \geq \p(k-1)\p(k+1)$. We write $\classlogconcave$ for the class of all log-concave distributions over $\Z$, and $\classlogconcave_n\subseteq \classlogconcave$ for that of all log-concave distributions over $[n]$.
\end{definition}

\paragraph{Discrete Fourier Transform}
For our SIIRV testing algorithm, we will need the following definition of the Fourier transform. 

\begin{definition}[Discrete Fourier Transform]
For $x \in \R$, we let $e(x) \eqdef  \exp(-2i\pi x)$. The \emph{Discrete Fourier Transform (DFT) modulo $M$} of a function
$F\colon[n] \to \C$ is  the function $\fourier{F}\colon[M]\to \C$ defined as
\[
    \fourier{F}(\xi)=\sum_{j=0}^{n-1} e\!\left(\frac{\xi j}{M}\right) F(j)
\]
for $\xi \in [M]$. The DFT modulo $M$ of a distribution $\p$, $\fourier{\p}$, is then the DFT modulo $M$ of its probability mass function (note that one can then equivalently see $\fourier{\p}(\xi)$ as the expectation $\fourier{\p}(\xi) = \E_{X\sim F}[e\!\left(\frac{\xi X}{M}\right)]$, for $\xi\in[M]$).

The \emph{inverse DFT modulo $M$} onto the range $[m,m+M-1]$ of $\fourier{F}\colon [M] \to \C$, is the function $F\colon [m, m+M-1] \cap \Z \to \C$ defined by 
\[
    F(j)= \frac{1}{M} \sum_{\xi=0}^{M-1} e\!\left(-\frac{\xi j}{M}\right) \fourier{F}(\xi),
\]
for $j \in [m, m+M-1] \cap \Z$.
\end{definition}

Note that the DFT (modulo $M$) is a linear operator; moreover, we recall the standard fact relating the norms of a function and of its Fourier transform, that we will use extensively:
\begin{theorem}[Plancherel's Theorem]
For $M\geq 1$ and $F,G\colon[n] \to \C$, we have (i) $\sum_{j=0}^{n-1} F(j)\overline{G(j)} =  \frac{1}{M}\sum_{\xi=0}^{M-1} \fourier{F}(\xi)\overline{\fourier{G}(\xi)}$; and (ii) $\normtwo{F}= \frac{1}{\sqrt{M}}\normtwo{\fourier{F}}$, 
where $\fourier{F},\fourier{G}$ are the DFT modulo $M$ of $F,G$, respectively.
\end{theorem}
\noindent(The latter equality is sometimes referred to as Parseval's theorem.) We also note that, for our PMD testing, we shall need the appropriate generalization of the Fourier transform to the multivariate setting. We leave this generalization to the corresponding section,~\cref{sec:pmd:testing}.

\paragraph{Tools from Probability}
We finally recall a classical inequality for sums of independent random variables, due to Bennett~\cite[Chapter 2]{Boucheron:13}:
\begin{theorem}[Bennett's inequality]
Let $X=\sum_{i=1}^n X_i$, where $X_1,\dots,X_n$ are independent random variables such that (i) $\E[X_i]=0$ and (ii) $\abs{X_i}\leq \alpha$ almost surely for all $1\leq i\leq n$. Letting $\sigma^2=\Var[X]$, we have, for every $t\geq 0$,
\[
    \Pr[ X > t ] \leq \exp\left( -\frac{\Var[X]}{\alpha^2} \vartheta\!\left( \frac{\alpha t}{\Var[X]} \right) \right)
\]
where $\vartheta(x)=(1+x)\ln(1+x) - x$.
\end{theorem}

\paragraph{Structural Results on SIIRVs}
To establish the completeness of our algorithms, we will rely on this lemma from~\cite{DKS:16}:
\begin{lemma}[{\cite[Lemma 2.3]{DKS:16}}]\label{lemma:FourierSupportLem}
Let $\p \in \classksiirv[n]{k}$ with $\sqrt{\Var_{X \sim \p}[X]} = s$, $1/2>\delta>0$, and $M \in \Z_+$ with $M>s$.
Let $\fourier{\p}$ be the discrete Fourier transform of $\p$ modulo $M$. Then, we have
  \begin{enumerate}
    \item[(i)]\label{lemma:FourierSupportLem:i} Let $\mathcal{L} = \mathcal{L}(\delta, M,s) \eqdef \left\{ \xi \in [M-1] \mid \exists a, b \in \Z, 0 \leq a \leq b < k \textrm{ such that }
    |\xi/M - a/b| <  \frac{\sqrt{\ln (1/\delta)}}{2s}  \right\} \;.$ Then, $|\fourier{\p}(\xi)| \leq \delta$ for all $\xi \in [M-1] \setminus \mathcal{L}.$
    That is, $|\fourier{\p}(\xi)| > \delta$ for
    at most $|\mathcal{L}| \leq M k^2 s^{-1} \sqrt{\log(1/\delta)}$ values of $\xi$ .
    \item[(ii)]\label{lemma:FourierSupportLem:ii} At most $4Mks^{-1}\sqrt{\log(1/\delta)}$ many integers $0 \leq \xi \leq M-1$ have  $|\fourier{\p}(\xi)| > \delta \;.$
  \end{enumerate}
\end{lemma}

We also provide a simple structural lemma, bounding the $L_2$ norm of any $(n,k)$-SIIRV as a function of $k$ and its variance only:
\begin{lemma}[Any $(n,k)$-SIIRV modulo $M$ has small $L_2$ norm]\label{claim:ksiirv:l2:norm}
  If $\p \in \mathcal{S}_{n, k}$ has variance $s^2$, then the distribution $\p'$ defined as $\p' \eqdef \p \bmod M$ satisfies 
  $
      \normtwo{\p'} \leq \sqrt{\frac{8k}{s}}
  $.
\end{lemma}
\noindent The proof of this lemma is deferred to~\cref{app:omitted}.

%% file: fouriertesting.tex
In this section, we prove the following theorem, which will be invoked as a crucial ingredient of our testing algorithms. Broadly speaking, the theorem ensures one can efficiently test whether an unknown distribution $\q$ has its Fourier transform concentrated on some (small) effective support $S$ (and if this is the case, learn \new{the vector} $\fourier{\q}\mathbf{1}_S$, the restriction of this Fourier transform to $S$, in $L_2$ distance).

\begin{theorem}\label{theo:ft:effective:support}
    Given parameters $M\geq 1$, $\eps,b\in(0,1]$, as well as a subset $S\subseteq [M]$ and sample access to a distribution $\q$ over $[M]$, \cref{algo:ft:effective:support} outputs either \reject or a collection of Fourier coefficients $\fourier{\h'}=(\fourier{\h'}(\xi))_{\xi\in S}$ such that with probability at least $7/10$, all the following statements hold simultaneously.
    \begin{enumerate}
        \item\label{theo:ft:effective:support:i} if $\normtwo{\q}^2 > 2b$, then it outputs \reject;
        \item\label{theo:ft:effective:support:ii} if $\normtwo{\q}^2 \leq 2b$ and every function $\q^\ast\colon[M]\to\R$ with $\fourier{\q^\ast}$ supported entirely on $S$ is such that $\normtwo{\q-\q^\ast} > \eps$, then it outputs \reject;
        \item\label{theo:ft:effective:support:iii} if $\normtwo{\q}^2 \leq b$ and there exists a function $\q^\ast\colon[M]\to\R$ with $\fourier{\q^\ast}$ supported entirely on $S$ such that $\normtwo{\q-\q^\ast} \leq \frac{\eps}{2}$, then it does not output \reject;
        \item\label{theo:ft:effective:support:iv} if it does not output \reject, then $\normtwo{\fourier{\q}\mathbf{1}_S-\fourier{\h'}} \leq \frac{\eps\sqrt{M}}{10}$ and the inverse Fourier transform (modulo $M$) $\h'$ of the Fourier coefficients $\fourier{\h'}$ it outputs satisfies $\normtwo{\q-\h'} \leq \frac{6\eps}{5}$.
    \end{enumerate}
    Moreover, the algorithm takes $m=\bigO{\frac{\sqrt{b}}{\eps^2}+ \frac{\abs{S}}{M\eps^2} +\sqrt{M}}$ samples from $\q$, and runs in time $\bigO{m\abs{S}}$.
\end{theorem}
Note that the rejection condition in~\cref{theo:ft:effective:support:ii} is equivalent to $\normtwo{\fourier{\q}\mathbf{1}_{\bar{S}}} > \eps\sqrt{M}$, that is to having Fourier mass more than $\eps^2$ outside of $S$; this is because for any $\q^\ast$ supported on $S$,
\[
    M\normtwo{\q-\q^\ast}^2 = \normtwo{\fourier{\q}-\fourier{\q^\ast}}^2
    = \normtwo{\fourier{\q}\mathbf{1}_{S}-\fourier{\q^\ast}\mathbf{1}_{S}}^2 + \normtwo{\fourier{\q}\mathbf{1}_{\bar{S}}-\fourier{\q^\ast}\mathbf{1}_{\bar{S}}}^2
    \geq \normtwo{\fourier{\q}\mathbf{1}_{\bar{S}}-\fourier{\q^\ast}\mathbf{1}_{\bar{S}}}^2
    = \normtwo{\fourier{\q}\mathbf{1}_{\bar{S}}}^2
\]
and the inequality is tight for $\q^\ast$ being the inverse Fourier transform (modulo $M$) of $\fourier{\q}\mathbf{1}_{S}$.

\medskip

\noindent {\bf High-level idea.} 
Let $\q$ be an unknown distribution supported on $M$ consecutive integers (we will later apply this to $\q\eqdef \p \bmod M$), and $S\subseteq[M]$ be a set of Fourier coefficients (symmetric with regard to $M$: $\xi\in S$ implies $-\xi \bmod M \in S$) such that $0\in S$. We can further assume that we know $b\geq 0$ such that $\normtwo{\q}^2 \leq b$.

Given $\q$, we can consider its ``truncated Fourier expansion'' (with respect to $S$) $\fourier{\h}=\hat{\q}\mathbf{1}_{S}$ defined as
\[
    \fourier{\h}(\xi) \eqdef
      \begin{cases}
          \hat{\q}(\xi) & \text{ if } \xi\in S\\
          0 & \text{ otherwise}
      \end{cases}
\]
for $\xi\in[M]$; and let $\h$ be the inverse Fourier transform (modulo $M$) of $\fourier{\h}$. Note that $\h$ is no longer in general a probability distribution.\medskip

To obtain the guarantees of~\cref{theo:ft:effective:support}, a natural idea is to take some number $m$ of samples from $\q$, and consider the empirical distribution $\q'$ they induce over $[M]$. By computing the Fourier coefficients (restricted to $S$) of this $\q'$, as well as the Fourier mass ``missed'' when doing so (i.e., the Fourier mass $\normtwo{\fourier{\q'}\mathbf{1}_{\bar{S}}}^2$ that $\q'$ puts outside of $S$) to sufficient accuracy, one may hope to prove~\cref{theo:ft:effective:support} with a reasonable bound on $m$.

The issue is that analyzing \emph{separately} the behavior of $\normtwo{\fourier{\q'}\mathbf{1}_{\bar{S}}}^2$ and $\normtwo{\fourier{\q'}\mathbf{1}_{S}-\fourier{\q'}\mathbf{1}_{S}}^2$ to show that they are both estimated sufficiently accurately, and both small enough, is not immediate. Instead, we will get a bound on both at the same time, by arguing concentration in a different manner -- namely, by analyzing a different tester for tolerant identity testing in $L_2$ norm.

In more detail, letting $\h$ be as above, we have by Plancherel that
\[
  \littlesum_{i\in [M]} (\q'(i)-\h(i))^2 = \normtwo{\q'-\h}^2 = \frac{1}{M}\normtwo{\fourier{\q'}-\fourier{\h}}^2 = \frac{1}{M}\littlesum_{\xi=0}^{M-1} \dabs{\fourier{\q'}(\xi)-\fourier{\h}(\xi)}^2
\]
and, expanding the definition of $\fourier{\h}$ and using Plancherel again, this can be rewritten as
\begin{align*}
  M\littlesum_{i\in [M]} (\q'(i)-\h(i))^2 &=  \normtwo{\fourier{\q}\mathbf{1}_{S}-\fourier{\q'}\mathbf{1}_{S}}^2 + \normtwo{\q'}^2 - \normtwo{\fourier{\q'}\mathbf{1}_{S}}^2.
\end{align*}
(The full derivation will be given in the proof.) The left-hand side has two non-negative compound terms: the first, $\normtwo{\fourier{\p}\mathbf{1}_{S}-\fourier{\q'}\mathbf{1}_{S}}^2$, corresponds to the $L_2$ error obtained when learning the Fourier coefficients of $\q$ on $S$. The second, $\normtwo{\q'}^2 - \normtwo{\fourier{\q'}\mathbf{1}_{S}}^2 = \normtwo{\fourier{\q'}\mathbf{1}_{\bar{S}}}^2$, is the Fourier mass that our empirical $\q'$ puts ``outside of $S$.''

So if the LHS is small (say, order $\eps^2$), then in particular both terms of the RHS will be small as well, effectively giving us bounds on our two quantities in one shot. But this very same LHS is very reminiscent of a known statistic~\cite{CDVV14} for testing identity of distributions in $L_2$. So, 
one can analyze the number of samples required by analyzing such an $L_2$ tester instead. 
This is what we will do in~\cref{prop:l2:identity:tester}.

\begin{algorithm}
  \begin{algorithmic}[1]
    \Require parameters $M\geq 1$, $b,\eps\in(0,1]$; set $S\subseteq [M]$; sample access to distribution $\q$ over $[M]$
    \State\label{algo:ft:step:choosemprime} Set $m\gets \clg{C(\frac{\sqrt{b}}{\eps^2}+ \frac{\abs{S}}{M\eps^2}+ \sqrt{M})}$ \Comment{$C>0$ is an absolute constant}    \State Draw $m'\gets \Poi(m)$; if $m'>2m$, \Return \reject
    \State\label{algo:ft:step:empr} Draw $m'$ samples from $\q$, and let $\q'$ be the corresponding empirical distribution over $[M]$
    \State\label{algo:ft:step:norm} Compute $\normtwo{\q'}^2$, $\fourier{\q'}(\xi)$ for every $\xi\in S$, and $\normtwo{\fourier{\q'}\mathbf{1}_S}^2$ \Comment{Takes time $\bigO{m\abs{S}}$}
    \If{ $m'^2\normtwo{\q'}^2 - m' > \frac{3}{2}bm^2$ }\label{algo:ft:step:norm:check} \Return \reject
    \ElsIf{ $\normtwo{\q'}^2 - \frac{1}{M}\normtwo{\fourier{\q'}\mathbf{1}_S}^2 \geq 3\eps^2\left(\frac{m'}{m}\right)^2+\frac{1}{m'}$ } \Return \reject
    \Else
      \State \Return $\fourier{\h'}=(\fourier{\q'}(\xi))_{\xi\in S}$
    \EndIf
  \end{algorithmic}
  \caption{Testing the Fourier Transform Effective Support}\label{algo:ft:effective:support}
\end{algorithm}

\begin{proof}[Proof of~\cref{theo:ft:effective:support}]
Given $m'\sim\Poi(m)$ samples from $\q$, let $\q'$ be the empirical distribution they define. 
We first observe that with probability $2^{-\Omega(\eps^2m/b)}< \frac{1}{100}$, we have $m' \in [1\pm \frac{\eps}{100\sqrt{b}}]m$ and thus the algorithm does not output \reject in Step~\ref{algo:ft:step:choosemprime} (this follows from standard concentration bounds on Poisson random variables). We will afterwards assume this holds. By Plancherel, we have
\[
  \sum_{i\in [M]} (\q'(i)-\h(i))^2 = \normtwo{\q'-\h}^2 =  \frac{1}{M}\normtwo{\fourier{\q'}-\fourier{\h}}^2 = \frac{1}{M}\sum_{\xi=0}^{M-1} \dabs{\fourier{\q'}(\xi)-\fourier{\h}(\xi)}^2
\]
and, expanding the definition of $\fourier{\h}$, this yields
\begin{align}
  \sum_{i\in [M]} (\q'(i)-\h(i))^2 &=  \frac{1}{M}\sum_{\xi\in S} \dabs{\fourier{\q'}(\xi)-\fourier{\h}(\xi)}^2 + \frac{1}{M}\sum_{\xi\notin S} \dabs{\fourier{\q'}(\xi)}^2 \notag\\
  &=  \frac{1}{M}\sum_{\xi\in S} \dabs{\fourier{\q'}(\xi)-\fourier{\q}(\xi)}^2 + \frac{1}{M}\sum_{\xi=0}^{M-1} \dabs{\fourier{\q'}(\xi)}^2 
  - \frac{1}{M}\sum_{\xi\in S} \dabs{\fourier{\q'}(\xi)}^2 \notag\\
  &=  \frac{1}{M}\left(\normtwo{\fourier{\q}\mathbf{1}_{S}-\fourier{\q'}\mathbf{1}_{S}}^2 + \normtwo{\fourier{\q'}}^2 - \normtwo{\fourier{\q'}\mathbf{1}_{S}}^2\right)  \notag\\
  &=  \frac{1}{M}\normtwo{\fourier{\q}\mathbf{1}_{S}-\fourier{\q'}\mathbf{1}_{S}}^2 + \normtwo{\q'}^2 - \frac{1}{M}\normtwo{\fourier{\q'}\mathbf{1}_{S}}^2  \label{eq:plancherel:statistic}
\end{align}
where in the last step we invoked  Plancherel again to argue that $\frac{1}{M}\normtwo{\fourier{\q'}}^2=\normtwo{\q'}^2$.

To analyze the correctness of the algorithm (specifically, the completeness), we will adopt the point of view suggested by~\eqref{eq:plancherel:statistic} and analyze instead the statistic
$\sum_{i\in [M]} (\q'(i)-\h(i))^2$, when $\h$ is an explicit (pseudo) distribution on $[M]$ assumed known, and $\q'$ is the empirical distribution obtained by drawing $\Poi(m)$ samples from some unknown distribution $\q$. (Namely, we want to see this as a tolerant $L_2$ identity tester between $\q$ and $\h$.)

\begin{itemize}
  \item We first show that, given that $m'=\bigOmega{\frac{\abs{S}}{M\eps^2}}$, with probability at least $\frac{99}{100}$ we have
    $
        \normtwo{\fourier{\q}\mathbf{1}_S-\fourier{\h'}} \leq \frac{\sqrt{M}\eps}{10}
    $.
	We note that $m'\fourier{\q'}(\xi)$ is an sum of $m'$ i.i.d. numbers each of absolute value $1$ and mean $\fourier{\q}(\xi)$ (which has absolute value less than $1$). If $X$ is one of these numbers, $|X-\fourier{\q}(\xi)| \leq 2$ with probability $1$ and so the variance of the real and imaginary parts of $X$ is at most $4$. Thus the variance of the real and imaginary  parts of $m'\fourier{\q'}(\xi)$ is at most $4m'$. Then we have
	$\E[|\fourier{\q}(\xi) - \fourier{\q'}(\xi)|^2]=\E[ (\Re (\fourier{\q}(\xi) - \fourier{\q'}(\xi)))^2 + (\Im ( \fourier{\q}(\xi) - \fourier{\q'}(\xi)))^2] \leq 8/m'$. Summing over $S$, using that $\q'$ and $\h'$ have the same Fourier coefficients there, yields
	
\[
    \E\left[ \sum_{\xi\in S} \abs{\fourier{\q}(\xi)-\fourier{\h'}(\xi)}^2 \right] 
	\leq \frac{8|S|}{m'}
    \leq \frac{M\eps^2}{10000}
\]
and by Markov's inequality we get 
$\Pr\left[ \normtwo{\fourier{\q}\mathbf{1}_S-\fourier{\h'}}^2 \leq \frac{M\eps^2}{100} \right] = \Pr\left[\sum_{\xi\in S} \abs{\fourier{\q}(\xi)-\fourier{\h'}(\xi)}^2 \leq \frac{M\eps^2}{100} \right] \geq \frac{1}{100}$,
 concluding the proof.

  \item Then, let us consider~\cref{theo:ft:effective:support:i}: assume $\normtwo{\q}^2 > 2b$, and set $X\eqdef m'^2\normtwo{\q'}^2 - m'$. Then, 
\[
  \E[ X ] = \sum_{i=1}^M \E[ m'^2\q'(i)^2 ] - \sum_{i=1}^M \E[ m'\q'(i) ] = \sum_{i=1}^M ( m\q(i)+m^2\q(i)^2) - \sum_{i=1}^M m\q(i)= m^2\normtwo{\q}^2
\]
since the $m'\q'(i)$ are distributed as $\Poi(m\q(i))$. As all $m'\q'(i)$'s are independent by Poissonization, we also have
\[
  \Var[ X ] = \sum_{i=1}^M \Var[ m'^2\q'(i)^2 - m'\q'(i) ] = \sum_{i=1}^M ( 2m^2\q(i)^2 + 4m^3\q(i)^3) 
  = 2m^2 \normtwo{\q}^2 + 4m^3 \norm{\q}_3^3
\]
and by Chebyshev,
\[
    \Pr[ X \leq \frac{3}{2}m^2b] \leq \Pr\left[ \abs{ X - \E[ X ] } > \frac{1}{4} \E[ X ] \right]
    \leq 16\frac{\Var[ X ]}{\E[ X ]^2}
    \leq \frac{32}{m^2\normtwo{\q}^2} + \frac{64 \norm{\q}_3^3}{m \normtwo{\q}^4 }
\]
Since $\q$ is supported on $[M]$, $\normtwo{\q}^2 \geq \frac{1}{M}$ and the first term is at most $\frac{32M}{m^2}$. The second term, by monotonicity of $\ell_p$-norms, is at most 
$\frac{64 \normtwo{\q}^3}{m \normtwo{\q}^4 } = \frac{48}{m \normtwo{\q} } \leq \frac{48\sqrt{M}}{m}$. The RHS is then at most $\frac{1}{100}$ for a large enough choice of $C>0$ in the definition of $m$. Thus, with probability at least $1-\frac{1}{100}$ we have $m'^2\normtwo{\q'}^2 - m' > \frac{3}{2}b$, and the algorithm outputs \reject in Step~\ref{algo:ft:step:norm:check}.

Moreover, if $\normtwo{\q}^2 \leq b$, then the same analysis shows that
\[
    \Pr[ X > \frac{3}{2}m^2b ] \leq \Pr\left[ \abs{ X - \E[ X ] } > \frac{1}{2}\E[ X ] \right]
    \leq 4\frac{\Var[ X ]}{\E[ X ]^2} \leq \frac{1}{100}
\]
and with probability at least $1-\frac{1}{100}$ the algorithm does not output \reject in Step~\ref{algo:ft:step:norm}.

  \item Turning now to~\cref{theo:ft:effective:support:ii,theo:ft:effective:support:iii,theo:ft:effective:support:iv}: we assume that the algorithm does not output \reject in Step~\ref{algo:ft:step:norm} (which by the above happens with probability $99/100$ if $\normtwo{\q}^2 \leq b$; and can be assumed without loss of generality otherwise, since we then want to argue that the algorithm \emph{does} reject at some point in that case).
  
    By the remark following the statement of the theorem, it is sufficient to show that the algorithm outputs \reject (with high probability) if $\normtwo{\fourier{\q}\mathbf{1}_{\bar{S}}}^2 > \eps^2 M$, and that if both $\normtwo{\q}^2 \leq b$ and $\normtwo{\fourier{\q}\mathbf{1}_{\bar{S}}}^2 \leq \frac{\eps^2}{4}M$ then it does not output~\reject; and that whenever the algorithm does not output \reject, then $\normtwo{\fourier{\q}-\fourier{\h}} \leq \eps^2 M$.
    
    Observe that calling~\cref{algo:tol:l2:identity:tester} with our $m'=\Poi(m)$ samples from $\q$ (distribution over $[M]$), parameters $\frac{\eps}{2}$ and $2b$, and the explicit description of the pseudo distribution $\p^\ast\eqdef \frac{m'}{m}\h$ (which one would obtain for $\h$ being the inverse Fourier transform of $\fourier{\q}\mathbf{1}_S$) would result by~\cref{prop:l2:identity:tester} (since $m \geq c\frac{\sqrt{2b}}{(\eps/2)^2} = 244\sqrt{2}\frac{\sqrt{b}}{\eps^2}$, where $c$ is as in~\cref{prop:l2:identity:tester}) in having the following guarantees on $\frac{\sqrt{Z}}{m}$, where $Z$ is the statistic defined in~\cref{algo:tol:l2:identity:tester}
    \begin{itemize}
      \item if $\normtwo{\q-\p^\ast} \leq \frac{\eps}{2}$, then $\frac{\sqrt{Z}}{m} \leq \sqrt{2.9}\eps$ with probability at least $3/4$;
      \item if $\normtwo{\q-\p^\ast} \geq \eps$, then $\frac{\sqrt{Z}}{m} \geq \sqrt{3.1}\eps$ with probability at least $3/4$;
    \end{itemize}
    as $\normtwo{\q}^2 \leq 2b$ (note that then $\normtwo{\h}^2 \leq b$ as well). Since $\sqrt{M}\normtwo{\q-\p^\ast} = \normtwo{\fourier{\q}-\fourier{\p^\ast}} = \normtwo{\fourier{\q}-\frac{m}{m'}\fourier{\q}\mathbf{1}_S}$ and     \[
      \frac{Z}{m'^2} = \sum_{i=1}^M \left( (\q'(i)-\frac{m}{m'}\p^\ast(i))^2 - \frac{\q'(i)}{m'} \right) = \sum_{i=1}^M (\q'(i)-\h(i))^2 - \frac{1}{m'}
    \]
    which is equal to $\frac{1}{M}\normtwo{\fourier{\q}\mathbf{1}_{S}-\fourier{\q'}\mathbf{1}_{S}}^2 + \normtwo{\q'}^2 - \frac{1}{M}\normtwo{\fourier{\q'}\mathbf{1}_{S}}^2 - \frac{1}{m'}$ by~\cref{eq:plancherel:statistic}, we thus get the following.
    \begin{itemize}
      \item if $\normtwo{\fourier{\q}\mathbf{1}_{\bar{S}}}^2 \leq \frac{\eps^2M}{9}$, then $\normtwo{\fourier{\q}-\fourier{\q}\mathbf{1}_{S}} \leq \frac{\eps}{3}\sqrt{M}$, and
      \[
          \sqrt{M}\normtwo{\p^\ast-\q} = \normtwo{\fourier{\p^\ast}-\fourier{\q}} \leq \normtwo{\fourier{\p^\ast}-\fourier{\q}\mathbf{1}_{S}} + \normtwo{\fourier{\q}\mathbf{1}_{S}-\fourier{\q}}
          = \abs{\frac{m}{m'}-1}\normtwo{\fourier{\q}\mathbf{1}_{S}} + \normtwo{\fourier{\q}-\fourier{\q}\mathbf{1}_{S}}
      \]
      Since we have $m'\in[1\pm\frac{\eps}{100\sqrt{b}}]m$ by the above discussion and $\normtwo{\fourier{\q}\mathbf{1}_{S}}\leq \sqrt{2b}\sqrt{M}$, the RHS is upper bounded by $\frac{\eps}{6}\sqrt{M} + \frac{\eps}{3}\sqrt{M} = \frac{\eps}{2}\sqrt{M}$, and $\normtwo{\p^\ast-\q} \leq \frac{\eps}{2}$.
      Then $\frac{1}{M}\normtwo{\fourier{\q}\mathbf{1}_{S}-\fourier{\q'}\mathbf{1}_{S}}^2 + \normtwo{\q'}^2 - \frac{1}{M}\normtwo{\fourier{\q'}\mathbf{1}_{S}}^2 = \frac{Z}{m'^2}+\frac{1}{m'} \leq 2.9\eps^2\left(\frac{m'}{m}\right)^2+\frac{1}{m'}$ with probability at least $3/4$, and in particular 
      $\normtwo{\q'}^2 - \frac{1}{M}\normtwo{\fourier{\q'}\mathbf{1}_{S}}^2 \leq 2.9\eps^2\left(\frac{m'}{m}\right)^2+\frac{1}{m'} < 3\eps^2\left(\frac{m'}{m}\right)^2+\frac{1}{m'}$;
      \item if $\normtwo{\fourier{\q}\mathbf{1}_{\bar{S}}}^2 > \eps^2M$,
           then $\frac{1}{M}\normtwo{\fourier{\q}\mathbf{1}_{S}-\fourier{\q'}\mathbf{1}_{S}}^2 + \normtwo{\q'}^2 - \frac{1}{M}\normtwo{\fourier{\q'}\mathbf{1}_{S}}^2 =\frac{Z}{m'^2}+\frac{1}{m'} > 3.1\eps^2\left(\frac{m'}{m}\right)^2+\frac{1}{m'}$ with probability at least $3/4$; since by the first part we established we have $\normtwo{\fourier{\q}\mathbf{1}_{S}-\fourier{\q'}\mathbf{1}_{S}}^2\leq \frac{\eps^2M}{100}$, this implies $\normtwo{\q'}^2 - \frac{1}{M}\normtwo{\fourier{\q'}\mathbf{1}_{S}}^2 > 3.1\eps^2\left(\frac{m'}{m}\right)^2+\frac{1}{m'} - \frac{\eps^2}{100}>  3\eps^2\left(\frac{m'}{m}\right)^2+\frac{1}{m'}$.
    \end{itemize}
    
    This immediately takes care of~\cref{theo:ft:effective:support:ii,theo:ft:effective:support:iii}; moreover, this implies that whenever \cref{algo:ft:effective:support} does \emph{not} output \reject, then the inverse Fourier transform $\h'$ of the collection of Fourier coefficients it returns (which are supported on $S$) satisfies
    \begin{align*}
       \normtwo{\q-\h'}^2 &= \frac{1}{M}\normtwo{\fourier{\q}-\fourier{\h'}}^2 = \frac{1}{M}\normtwo{\fourier{\q}\mathbf{1}_{S}-\fourier{\h'}}^2 + \frac{1}{M}\normtwo{\fourier{\q}\mathbf{1}_{\bar{S}}}^2 \\
        &\leq \frac{\eps^2}{100} + \frac{1}{M}\normtwo{\fourier{\q}\mathbf{1}_{\bar{S}}}^2 \\
        &\leq \frac{\eps^2}{100} + \eps^2 = \frac{101}{100}\eps^2
    \end{align*}
    and thus
    $
        \normtwo{\q-\h'} \leq  \sqrt{\frac{101}{100}}\eps < \frac{6}{5}\eps
    $
    which establishes~\cref{theo:ft:effective:support:iv}. Finally, by a union bound, all the above holds except with probability $\frac{1}{100}+\frac{1}{100}+\frac{1}{100}+\frac{1}{4}<\frac{3}{10}$. This concludes the proof.

\end{itemize}

\end{proof}

\subsection{A Tolerant $L_2$ Tester for Identity to a Pseudodistribution}

As previously mentioned, one building block in the proof of~\cref{theo:ft:effective:support} (and a result that may be of independent interest) is an optimal $L_2$ identity testing algorithm. Our tester and its analysis are very similar to the tolerant $L_2$ closeness testing algorithm of Chan et al.~\cite{CDVV14}, with the obvious simplifications pertaining to identity (instead of closeness). The main difference is that we emphasize here the fact that $\p^\ast$ need not be an actual distribution: any $\p^\ast\colon[r]\to\R$ would do, even taking negative values. This will turn out to be crucial for our applications.

\begin{algorithm}
  \begin{algorithmic}
  \Require $\eps\in(0,1)$, $\Poi(m)$ samples from distributions $\p$ over $[r]$, with $X_i$ denoting the number of occurrences of the $i$-th domain elements in the samples from $\p$, and $\p^\ast$ being a fixed, known pseudo distribution over $[r]$.
  \Ensure Returns \textsf{accept} if $\normtwo{\p - \p^\ast } \leq \eps$ and \textsf{reject} if $\normtwo{\p - \p^\ast }  \geq 2\eps$.
    \State Define $Z=\sum_{i=1}^r (X_i-m\p^\ast(i))^2-X_i.$ \Comment{Can actually be computed in $O(m)$ time}
    \State Return  \textsf{reject} if $\frac{\sqrt{Z}}{m} > \sqrt{3}\eps$, \textsf{accept} otherwise.
  \end{algorithmic}
  \caption{Tolerant $L_2$ identity tester}\label{algo:tol:l2:identity:tester}
\end{algorithm}

\begin{proposition}\label{prop:l2:identity:tester}
There exists an absolute constant $c > 0$ such that the above algorithm (\cref{algo:tol:l2:identity:tester}), when given $\Poi(m)$ samples drawn from a distribution $\p$ and an explicit function $\p^\ast\colon[r]\to\R$ will, with probability at least $3/4$, 
distinguishes between \textsf{(a)} $\normtwo{ \p-\p^\ast } \leq \eps$ and \textsf{(b)} $\normtwo{ \p-\p^\ast } \geq 2\eps$ provided that $m \geq c\frac{\sqrt{b}}{\eps^2}$, where $b$ is an upper bound on $\normtwo{ \p }^2, \normtwo{ \p^\ast }^2$. (Moreover, one can take $c = 61$.)

Moreover, we have the following stronger statement: in case (a), the statistic $Z$ computed in the algorithm satisfies $\frac{\sqrt{Z}}{m} \leq \sqrt{2.9}\eps$ with probability at least $3/4$, while in case (b) we have $\frac{\sqrt{Z}}{m} \geq \sqrt{3.1}\eps$ with probability at least $3/4$.
\end{proposition}

%% file: siirvtesting.tex
We are now ready to describe the algorithm behind~\cref{theo:testing:ksiirv}, and establish the theorem.
\begin{algorithm}[ht]
  \begin{algorithmic}[1]
    \Require sample access to a distribution $\p\in\distribs{\N}$, parameters $n,k\geq 1$ and $\eps\in(0,1]$
    
    \State\Comment{ Let $C,C',C''$ be sufficiently large universal constants }
    
    \State\label{algo:step:estimates:mu:sigma} Draw $O(k)$ samples from $\p$ and compute as in~\cref{claim:estimate:moments}: (a) $\widetilde{\sigma}^2$, a tentative factor-$2$ approximation to $\Var_{X \sim \p}[X]+1$,
and (b) $\widetilde{\mu}$, a tentative approximation to $\E_{X \sim \p}[X]$ to within one standard deviation.
    \If{If $\widetilde{\sigma} > 2k\sqrt{n}$} \label{algo:step:check:stdev}
        \State \Return \reject    \Comment{Blatant violation of $(n,k)$-SIIRV-iness}
    \EndIf
    
    \If{ $\widetilde{\sigma}  \leq 2k\sqrt{\ln\frac{10}{\eps}}$ }\label{algo:step:small:variance}
        \State\label{algo:step:def:interval:smallvariance} Set $M \gets 1+2\clg{15k \ln\frac{10}{\eps}}$, and let $I \gets [\flr{ \widetilde{\mu} }-\frac{M-1}{2},\flr{ \widetilde{\mu} }+\frac{M-1}{2}]$; and $S\gets [M]$
        \State\label{algo:step:effectivesupport:smallvariance} Draw $O(1/\eps)$ samples from $\p$, to distinguish between $\p(I) \leq 1-\frac{\eps}{4}$ and $\p(I) > 1-\frac{\eps}{5}$.
        If the former is detected, \Return \reject
        \State\label{algo:step:empirical:smallvariance} Take $N=C\left( \frac{\abs{S}}{\eps^2} \right) = \bigO{ \frac{k}{\eps^2} \log\frac{1}{\eps} }$ samples from $\p$ to get an empirical distribution $\h$
    \Else\label{algo:step:big:variance}
        \State\label{algo:step:def:interval} Set $M \gets 1+2\clg{ 4\widetilde{\sigma}\sqrt{\ln(4/\eps)} }$, and let $I \gets [\flr{ \widetilde{\mu} }-\frac{M-1}{2},\flr{ \widetilde{\mu} }+\frac{M-1}{2}]$
        \State\label{algo:step:effectivesupport} Draw $O(1/\eps)$ samples from $\p$, to distinguish between $\p(I) \leq 1-\frac{\eps}{4}$ and $\p(I) > 1-\frac{\eps}{5}$.
        If the former is detected, \Return \reject
        \State\label{algo:step:dft:computation} Let $\delta\gets \frac{\eps}{C''\sqrt{k\log\frac{k}{\eps}}}$, and 
            \[
              S \gets \setOfSuchThat{ \xi \in [M-1] }{ \exists a, b \in \Z, 0 \leq a \leq b < k \textrm{ s.t. } |\xi/M - a/b| \leq  C'\frac{\sqrt{\ln(1/\delta)}}{4\widetilde{\sigma}} } \;.
            \]
        \State\label{algo:step:fourier:support} Simulating sample access to $\p'\eqdef \p\bmod M$, call~\cref{algo:ft:effective:support} on $\p'$ with parameters $M$, $\frac{\eps}{5\sqrt{M}}$, $b=\frac{16k}{\widetilde{\sigma}}$, and $S$. If it outputs \reject, then \Return \reject; otherwise, let $\fourier{\h}=(\fourier{\h}(\xi))_{\xi\in S}$ denote the collection of Fourier coefficients it outputs, and $\h$ their inverse Fourier transform (modulo $M$) \Comment{Do not actually compute $\h$}
    \EndIf
    \State \label{algo:step:cover}\textsf{Projection Step:} Check whether $\dtv(\h,\classksiirv[n]{k}) \leq \frac{\eps}{2}$ (as in~\cref{sec:projections}), and \Return \accept if it is the case. If not, \Return \reject.  \Comment{Mostly computational step}
  \end{algorithmic}
  \caption{Algorithm \texttt{Test-SIIRV}}\label{algo:ksiirv:tester}
\end{algorithm}

\subsection{Analyzing the Subroutines}

We start with a simple fact, that we will use to bound the running time of our algorithm \new{and which follows immediately from~\cite[Claim 2.4]{DKS:16}:}
\begin{fact}\label{fact:bound:size:s}
  For $S$ as defined in Step~\ref{algo:step:dft:computation}, we have
  \[
  \abs{S} \leq Mk^2 \frac{C'}{2\widetilde{\sigma}} \sqrt{\ln\frac{1}{\delta}} \leq 100C' k^2 \sqrt{\ln\frac{4}{\eps}}\sqrt{\ln \frac{k}{\eps} +\log\log \frac{k}{\eps} + \frac{1}{2} \ln(16C'')}
  \leq C'' k^2\log^2\frac{k}{\eps}
  \]
  for a suitably large choice of the constant $C''>0$; from which we get $\delta \leq \frac{1}{4\sqrt{\abs{S}}}$.
\end{fact}

We then argue that with high probability, the estimates obtained in Step~\ref{algo:step:estimates:mu:sigma} will be accurate enough for our purposes. (The somewhat odd statement below, stating two distinct guarantees where the second implies the first, is due to the following: \cref{eq:guarantees:moments} will be the guarantee that (the completeness analysis of) our algorithm relies on, while the second, slightly stronger one, will only be used in the particular implementation of the ``projection step'' (Step~\ref{algo:step:cover}) from~\cref{sec:projections}.)
\begin{claim}[Estimating the first two moments (if $\p$ is a SIIRV)]\label{claim:estimate:moments}
  With probability at least $19/20$ over the $O(k)$ draws from $\p$ in Step~\ref{algo:step:estimates:mu:sigma}, the following holds. If $\p\in\classksiirv[n]{k}$, the estimates $\widetilde{\sigma},\widetilde{\mu}$ defined as the empirical mean and (unbiased) empirical variance meet the guarantees stated in Step~\ref{algo:step:estimates:mu:sigma} of the algorithm, namely
  \begin{equation}\label{eq:guarantees:moments}
      \frac{1}{2}\leq \frac{\widetilde{\sigma}^2}{\Var_{X \sim \p}[X]+1} \leq 2, \qquad \abs{ \widetilde{\mu} - \E_{X \sim \p}[X] } \leq  \sqrt{\Var_{X \sim \p}[X]}
  \end{equation}
  We even have a quantitatively slightly stronger guarantee:
  $
      \frac{2}{3}\leq \frac{\widetilde{\sigma}^2}{\Var_{X \sim \p}[X]+1} \leq \frac{3}{2}$, and $\abs{ \widetilde{\mu} - \E_{X \sim \p}[X] } \leq  \frac{1}{2}\sqrt{\Var_{X \sim \p}[X]}
  $.
\end{claim}
\begin{proof}We handle the estimation of the mean and variance separately.
\begin{description}
  \item[Estimating the mean.] $\widetilde{\mu}$ will be the usual empirical estimator, namely $\widetilde{\mu} \eqdef \frac{1}{m}\sum_{i=1}^m X_i$ for $X_1,\dots,X_m$ independently drawn from $\p$. Since $\E[\widetilde{\mu}]=\E_{X \sim \p}[X]$ and $\Var[\widetilde{\mu}]= \frac{1}{m}\Var_{X \sim \p}[X]$, Chebyshev's inequality guarantees that
  \[
      \Pr[ \abs{ \widetilde{\mu} - \E_{X \sim \p}[X] } >  \frac{1}{2}\sqrt{\Var_{X \sim \p}[X]}] \leq \frac{4}{m}
  \]
  which can be made at most $1/200$ by choosing $m\geq 800$.
  \item[Estimating the variance.] The variance estimation is exactly the same as in~\cite[Lemma 6]{DDS15-journal}, observing that their argument only requires that $\p$ be the distribution of a sum of independent random variables (not necessarily a Poisson Binomial distribution). Namely, they establish that,\footnote{\cite[Lemma 6]{DDS15-journal} actually only deals with the case $k=2$; but the bound we state follows immediately from their proof and the simple observation that the excess kurtosis $\kappa$ of an $(n,k)$-SIIRV with variance $s^2$ is at most ${k^2}/{s^2}$.} letting $\widetilde{\sigma}^2\eqdef\frac{1}{m-1}\sum_{i=1}^m (X_i-\frac{1}{m}\sum_{j=1}^m X_j)^2$ be the (unbiased) sample variances, and $s^2\eqdef \Var_{X \sim \p}[X]$,
   \[
       \Pr[ \abs{ \widetilde{\sigma}^2 - s^2 } > \alpha(1+s^2) ] \leq \frac{4s^4+k^2s^2}{\alpha^2(1+s^2)^2} \frac{1}{m}
       \leq \frac{4s^4+s^2}{\alpha^2(1+s^2)^2}\cdot  \frac{k^2}{m} \leq \frac{4k^2}{\alpha^2m}
   \]
   which for $\alpha=1/3$ is at most $9/200$ by choosing $m\geq 800k$.
\end{description}
A union bound completes the proof, giving a probability of error at most $\frac{1}{200}+\frac{9}{200}=\frac{1}{20}$.
\end{proof}

\begin{claim}[Checking the effective support]\label{claim:estimate:effectivesupport}
  With probability at least $19/20$ over the draws from $\p$ in Step~\ref{algo:step:effectivesupport}, the following holds.
  \begin{itemize}
    \item if $\p\in\classksiirv[n]{k}$ and~\eqref{eq:guarantees:moments} holds, then $\p(I)\geq 1-\frac{\eps}{5}$ and the algorithm does not output \reject in Step~\ref{algo:step:effectivesupport:smallvariance} nor~\ref{algo:step:effectivesupport};
    \item if $\p$ puts probability mass more than $\frac{\eps}{4}$ outside of $I$, then the algorithm outputs \reject in Step~\ref{algo:step:effectivesupport:smallvariance} or~\ref{algo:step:effectivesupport}.
  \end{itemize}
\end{claim}
\begin{proof}Suppose first $\p\in\classksiirv[n]{k}$ and~\eqref{eq:guarantees:moments} holds, and set $s\eqdef \sqrt{\Var_{X\sim\p}[X]}$ and $\mu\eqdef \E_{X\sim\p}[X]$ as before. By Bennett's inequality applied to $X$, we have
  \begin{equation}\label{eq:bennett:application}
    \Pr[ X > \mu+t ] \leq \exp\left(- \frac{s^2}{k^2} \vartheta\left( \frac{kt}{s^2} \right) \right)
  \end{equation}
  for any $t>0$, where $\vartheta\colon\mathbb{R}_+^\ast\to\mathbb{R}$ is defined by $\vartheta(x) = (1+x)\ln(1+x)-x$.
\begin{itemize}
  \item If the algorithm reaches Step~\ref{algo:step:effectivesupport:smallvariance}, then $s\leq 4k \sqrt{\ln\frac{10}{\eps}}$. Setting $t=\alpha\cdot k \ln\frac{10}{\eps}$ in~\cref{eq:bennett:application} (for $\alpha>0$ to be determined shortly), and $u = \frac{kt}{s^2} = \alpha\frac{k^2}{s^2}\ln\frac{10}{\eps} \geq \frac{\alpha}{16}$, 
  \[
      \frac{s^2}{k^2} \vartheta\left( \frac{kt}{s^2} \right) = \alpha\ln\frac{10}{\eps} \cdot \frac{\vartheta\left( u \right)}{u}
      \geq \left(16\vartheta\left( \frac{\alpha}{16} \right)\right) \ln\frac{10}{\eps} \geq \ln\frac{10}{\eps}
  \]
  since $\frac{\vartheta\left(x\right)}{x} \geq \frac{\vartheta\left(\alpha/16\right)}{\alpha/16}$ for all $x\geq \frac{\alpha}{16}$; the last inequality for $\alpha\geq \alpha^\ast \simeq 2.08$ chosen to be the solution to $16\vartheta\left( \frac{\alpha^\ast}{16} \right)=1$. Thus, $\Pr[X > \mu+t] \leq \frac{\eps}{10}$. Similarly, we have $\Pr[ X < \mu-t ] \leq \frac{\eps}{10}$.
  As $\mu-2t\leq\mu-s\leq \widetilde{\mu} \leq \mu+s \leq \mu +2t$, we get $\Pr[ X \in I ] \geq 1-\frac{\eps}{5}$ as claimed.

  \item If the algorithm reaches Step~\ref{algo:step:effectivesupport}, then $s\geq k \sqrt{\ln\frac{10}{\eps}}$ and $M = 1+2\clg{ 6\widetilde{\sigma}\sqrt{\ln\frac{10}{\eps}}) } \geq 1+2\clg{ 3s\sqrt{\ln\frac{10}{\eps}}) }$. Setting $t=\beta s \sqrt{ \ln\frac{10}{\eps}}$ in~\cref{eq:bennett:application} (for $\beta>0$ to be determined shortly), and $u = \frac{kt}{s^2} = \beta\frac{k}{s}\sqrt{ \ln\frac{10}{\eps}}\leq \beta$, 
  \[
      \frac{s^2}{k^2} \vartheta\left( \frac{kt}{s^2} \right) = \frac{t^2}{s^2} \cdot \frac{\vartheta\left( u \right)}{u^2}
      = \beta^2\ln\frac{10}{\eps}\cdot \frac{\vartheta\left(u\right)}{u^2} \geq \ln\frac{10}{\eps}
  \]
  since $\frac{\vartheta\left(x\right)}{x^2} \geq \frac{\vartheta\left(\beta\right)}{\beta^2}$ for all $x\in(0,\beta]$; the last inequality for $\beta= e-1 \simeq 1.72$ chosen to be the solution to $\vartheta\left( \beta \right)=1$. Thus, $\Pr[X > \mu+t] \leq \frac{\eps}{10}$. Similarly, it holds $\Pr[ X < \mu-t ] \leq \frac{\eps}{10}.$
Now note that $ \lfloor\widetilde{\mu} \rfloor+(M-1)/2 \geq (\mu-s)+\lceil 2s\sqrt{\ln\frac{10}{\eps}}) \rceil \geq \mu + t$
and $\flr{ \widetilde{\mu} } - (M-1)/2 \leq \mu-t$, implying that $X$ is in $[\flr{ \widetilde{\mu} }-(M-1)/2, \flr{ \widetilde{\mu} }+(M-1)/2]$ with probability at least $1- \frac{\eps}{5}$ as desired.
\end{itemize}

To conclude and establish the conclusion of the first item, as well as the second item, recall that distinguishing with probability $19/20$ between the cases $\p(\bar{I})\leq\frac{\eps}{5}$ and $\p(\bar{I})>\frac{\eps}{4}$ can be done with $O(1/\eps)$ samples.
\end{proof}

\begin{claim}[Learning when the effective support is small]\label{claim:learn:empirical}
  If $\p$ satisfies $\p(I)\geq 1-\frac{\eps}{4}$, and the ``If'' statement at Step~\ref{algo:step:small:variance} holds, then with probability at least $19/20$ the empirical distribution $\h$ obtained in Step~\ref{algo:step:empirical:smallvariance} satisfies (i) $\dtv(\p,\h) \leq \frac{\eps}{2}$ and (ii) $\normtwo{\fourier{\p}-\fourier{\h}} \leq \frac{\eps^2}{100}$.
\end{claim}
\begin{proof}The first item, (i), follows from standard bounds on the rate of convergence of the empirical distribution (namely, that $O(r/\eps^2)$ samples suffice for it to approximate an arbitrary distribution over support of size $r$ up to total variation distance $\eps$). Recalling that in this branch of the algorithm, $S=[M]$ with $M=O(k\log(1/\eps))$, the second item, (ii), is proven by the same argument as in (the first bullet in) the proof of~\cref{theo:ft:effective:support}.
\end{proof}

\begin{claim}[Any $(n,k)$-SIIRV puts near all its Fourier mass in $S$]\label{claim:ksiirv:fourier:concentrated}
  If $\p \in \classksiirv[n]{k}$ and~\eqref{eq:guarantees:moments} holds, then
  $
      \normtwo{\fourier{\p}\mathbf{1}_{\bar{S}}}^2 = \sum_{\xi\notin S} \dabs{\fourier{\p}(\xi)}^2 \leq \frac{\eps^2}{100}
  $.
\end{claim}
\begin{proof}Since $\p \in \classksiirv[n]{k}$, our assumptions imply that (with the notations of~\cref{lemma:FourierSupportLem}) the set of large Fourier coefficients satisfies $\setOfSuchThat{\xi\in[M-1]}{ \abs{\fourier{\p}(\xi)} > \delta } \subseteq \mathcal{L}(\delta, M, s)\subseteq S$. Therefore, $\xi\notin S$ implies $\dabs{\fourier{\p}(\xi)} \leq \delta$. We then can conclude as follows: applying~\cref{lemma:FourierSupportLem} (ii) with parameter $\delta 2^{-r-1}$ for each $r \geq 0$, this is at most
\begin{align}
\sum_{r\geq 0} (\delta 2^{-r})^2  \abs{\setOfSuchThat{ \xi }{ |\fourier{\p}(\xi)| > \delta 2^{-r-1} }} 
&\leq  \frac{4Mk \delta^2}{s} \sum_{r\geq 0} 4^{-r}\sqrt{\log(2^{r+2}/\delta)} \notag\\
& \leq \frac{4Mk \delta^2}{s} \sqrt{\log\frac{1}{\delta}} \sum_{r\geq 0} 4^{-r}\sqrt{\log(2^{r+1})} \notag\\
& \leq \frac{12Mk \delta^2}{s} \sqrt{\log\frac{1}{\delta}} = \bigO{\eps^2} \label{eq:dealing:small:coeffs:p}
\end{align}
again at most $\frac{\eps^2}{100}$ for big enough $C''$ in the definition of $\delta$.
\end{proof}

\subsection{Putting It Together}

In what follows, we implicitly assume that $I$ (as defined in Step~\ref{algo:step:def:interval} of~\cref{algo:ksiirv:tester}) is equal to $[M]$. This can be done without loss of generality, as this is just a shifting of the interval and all our Fourier arguments are made modulo $M$.

\begin{lemma}[Putting it together: completeness]\label{lemma:completeness}
  If $\p \in \classksiirv[n]{k}$, then the algorithm outputs \accept with probability at least $3/5$.
\end{lemma}
\begin{proof}
Assume $\p \in \classksiirv[n]{k}$. We condition on the estimates obtained in Step~\ref{algo:step:estimates:mu:sigma} to meet their accuracy guarantees, which by~\cref{claim:estimate:moments} holds with probability at least $19/20$: that is, we hereafter assume~\cref{eq:guarantees:moments} holds. Since the variance of any $(n,k)$-SIIRV is at most $s^2 \leq n k^2$, we consequently have $\widetilde{\sigma} \leq 2k\sqrt{n}$ and the algorithm does not output \reject in Step~\ref{algo:step:check:stdev}.

\begin{itemize}
\item \textbf{Case 1:} the branch in Step~\ref{algo:step:small:variance} is taken.
In this case, by~\cref{claim:estimate:effectivesupport} the algorithm does not output \reject in Step~\ref{algo:step:effectivesupport:smallvariance} with probability $19/20$. Since $\p(I)\geq 1-\frac{\eps}{4}$, by~\cref{claim:learn:empirical} we get that with probability at least $19/20$ it is the case that $\dtv(\p,\h)\leq \frac{\eps}{2}$, and therefore the computational check in Step~\ref{algo:step:cover} will succeed, and return \accept. Overall, by a union bound the algorithm is successful with probability at least $1-3/20>3/5$.

\item\textbf{Case 2:} the branch in Step~\ref{algo:step:big:variance} is taken.
In this case, by~\cref{claim:estimate:effectivesupport} the algorithm does not output \reject in Step~\ref{algo:step:effectivesupport} with probability $19/20$. From~\cref{claim:ksiirv:l2:norm}, we know that $\p'$ as defined in Step~\ref{algo:step:fourier:support} satisfies $\normtwo{\p'}^2 \leq \frac{8k}{s}\leq \frac{16k}{\widetilde{\sigma}} = b$.  Moreover,~\cref{claim:ksiirv:fourier:concentrated} guarantees that $\normtwo{\fourier{\p'}\mathbf{1}_{\bar{S}}} \leq \frac{\eps}{10\sqrt{M}} = \frac{\eps'}{2}$ (for $\eps' = \frac{\eps}{5\sqrt{M}}$). Since Step~\ref{algo:step:fourier:support} calls~\cref{algo:ft:effective:support} with parameters $M, \eps', b$, and $S$, \cref{theo:ft:effective:support:iii} of~\cref{theo:ft:effective:support} ensures that (with probability at least $7/10$) the algorithm will not output \reject in Step~\ref{algo:step:fourier:support}, but instead return the $S$-sparse Fourier transform of some $\h$ supported on $[M]$ with $\normtwo{\p'-\h} \leq \frac{6}{5}\eps' = \frac{6\eps}{25\sqrt{M}}$. 

By Cauchy--Schwarz, we then have $\normone{\p'-\h} \leq \sqrt{M}\normtwo{\p'-\h} \leq \frac{6\eps}{25}$, i.e. $\dtv(\p',\h) \leq \frac{3\eps}{25}$. But since
$\dtv(\p,\p') \leq \frac{\eps}{4}$, we get $\dtv(\p,\h) \leq \frac{\eps}{4}+\frac{3\eps}{25} < \frac{\eps}{2}$, and the computational check in Step~\ref{algo:step:cover} will succeed, and return \accept.
Overall, by a union bound the algorithm accepts with probability at least $1-(1/20+1/20+3/10)=3/5$.
\end{itemize}
\end{proof}

\begin{lemma}[Putting it together: soundness]\label{lemma:soundness}
  If $\dtv(\p,\classksiirv[n]{k}) > \eps$, then the algorithm outputs \reject with probability at least $3/5$.
\end{lemma}
\begin{proof}
We will proceed by contrapositive, and show that if the algorithm returns \accept with probability at least $3/5$ then $\dtv(\p,\classksiirv[n]{k}) \leq \eps$. 
Depending on the branch of the algorithm followed, we assume the samples taken either in
 \begin{itemize}
    \item Steps~\ref{algo:step:estimates:mu:sigma},~\ref{algo:step:effectivesupport:smallvariance},~\ref{algo:step:empirical:smallvariance}, meet the guarantees of~\cref{claim:learn:empirical,claim:estimate:moments,claim:estimate:effectivesupport} (by a union bound, this happens with probability at least $1-3/20 > 2/3$); or
    \item Steps~\ref{algo:step:estimates:mu:sigma},~\ref{algo:step:effectivesupport},~\ref{algo:step:fourier:support} meet the guarantees of~\cref{claim:estimate:moments,claim:estimate:effectivesupport,theo:ft:effective:support} (by a union bound, this happens with probability at least $1-(1/20+1/20+3/10)=3/5$).
 \end{itemize}
In particular, we hereafter assume that $\widetilde{\sigma} \leq 2k\sqrt{n}$.

\begin{itemize}
\item \textbf{Case 1:} the branch in Step~\ref{algo:step:small:variance} is taken.

  By the above discussion, we have $\p(I)\geq 1 - \frac{\eps}{4}$ by~\cref{claim:estimate:effectivesupport} so~\cref{claim:learn:empirical} and our conditioning ensure that the empirical distribution $\h$ is such that $\dtv(\p,\h) \leq \frac{\eps}{2}$. Since the algorithm did not reject in Step~\ref{algo:step:cover}, there exists a $(n,k)$-SIIRV $\p^\ast$ such that $\dtv(\h,\p^\ast) \leq \frac{\eps}{2}$: by the triangle inequality, $\dtv(\p,\classksiirv[n]{k}) \leq \dtv(\p,\q^\ast) \leq \eps$.

\item \textbf{Case 2:} the branch in Step~\ref{algo:step:big:variance} is taken.

  In this case, we have $\p(I)\geq 1 - \frac{\eps}{4}$ by~\cref{claim:estimate:effectivesupport}. Furthermore, as the algorithm did not output \reject on Step~\ref{algo:step:fourier:support}, by~\cref{theo:ft:effective:support} we know that the inverse Fourier transform (modulo $M$) $\h$ of the $S$-sparse collection of Fourier coefficients $\fourier{\h}$ returned satisfies
  $
       \normtwo{\h-\p'} \leq \frac{6\eps}{25\sqrt{M}}
  $ 
  which by Cauchy--Schwarz implies, as both $\h$ and $\p'$ are supported on $[M]$, that $\normone{\h-\p'}\leq \frac{6\eps}{25}$, or equivalently $\dtv(\h,\p')\leq \frac{3\eps}{25}$.
  
  Finally, since the algorithm outputted \accept in Step~\ref{algo:step:cover}, there exists $\p^\ast\in\classksiirv[n]{k}$ (supported on $[M]$) such that $\dtv(\h,\p^\ast)\leq \frac{\eps}{2}$, and by the triangle inequality
  \[
      \dtv(\p,\p^\ast) \leq \dtv(\p,\p') + \dtv(\h,\p') + \dtv(\h,\p^\ast) \leq \frac{\eps}{4} + \frac{3\eps}{25} + \frac{\eps}{2} \leq \eps
  \]
  and thus $\dtv(\p,\classksiirv[n]{k}) \leq \dtv(\p,\p^\ast) \leq \eps$.
 \end{itemize}
\end{proof}

\begin{lemma}[Sample complexity]\label{lemma:sample:complexity}
  The algorithm has sample complexity $O\left( \frac{k n^{1/4}}{\eps^2}\log^{1/4}\frac{1}{\eps} + \frac{k^2}{\eps^2} \log^2\frac{k}{\eps} \right)$.
\end{lemma}
\begin{proof}  \cref{algo:ksiirv:tester} takes samples in Steps~\ref{algo:step:estimates:mu:sigma},~\ref{algo:step:effectivesupport:smallvariance},~\ref{algo:step:effectivesupport}, and~\ref{algo:step:fourier:support}. The sample complexity is dominated by Steps~\ref{algo:step:empirical:smallvariance} and~\ref{algo:step:fourier:support}, which take respectively $N$ and 
  \begin{align*}
    O\left( \frac{\sqrt{b}}{(\eps/\sqrt{M})^2} + \frac{\abs{S}}{M(\eps/\sqrt{M})^2} + \sqrt{M} \right) 
    &= O\left( \frac{\sqrt{k \widetilde{\sigma}}}{\eps^2}\sqrt[4]{\log\frac{1}{\eps}} 
        + \frac{\abs{S}}{\eps^2}
        + \sqrt{\widetilde{\sigma}}\sqrt[4]{\log\frac{1}{\eps}} \right)
    \\
    &= O\left( \frac{k n^{1/4}}{\eps^2}\log^{1/4}\frac{1}{\eps} + \frac{k^2}{\eps^2}\log^2\frac{k}{\eps} \right)
  \end{align*}
  samples; recalling that Step~\ref{algo:step:check:stdev} ensured that $\widetilde{\sigma} \leq 2k\sqrt{n}$ and that $\abs{S}=\bigO{k^2\log^2\frac{k}{\eps}}$ by~\cref{fact:bound:size:s}.
\end{proof}

\begin{lemma}[Time complexity]\label{lemma:running:time}
  The algorithm runs in time $\bigO{ \frac{k^4n^{1/4}}{\eps^2}\log^{4}\frac{k}{\eps} } + T(n,k,\eps)$, where $T(n,k,\eps)=n(k/\eps)^{\bigO{k\log(k/\eps)}}$ is the running time of the projection subroutine of Step~\ref{algo:step:cover}.
\end{lemma}
\begin{proof}  The running time, depending on the branch taken, is either $O(N)+T(n,k,\eps)$ for the first or $O\left( \abs{S}\left(\frac{k n^{1/4}}{\eps^2}\log^{1/4}\frac{1}{\eps} + \frac{k^2}{\eps^2}\log^2\frac{k}{\eps}\right) \right)+T(n,k,\eps)$ for the second (the latter from the running time of~\cref{algo:ft:effective:support}). Recalling that $\abs{S}=\bigO{k^2\log^2\frac{k}{\eps}}$ by~\cref{fact:bound:size:s} yields the claimed running time.
\end{proof}

\subsection{The Projection Subroutines}\label{sec:projections}
\input{sec-projections}

%% file: sec-projections.tex
\subsubsection{The Projection Step for $(n,k)$-SIIRVs}

We can use the proper $\eps$-cover given in \cite{DKS:16} to find a $(n,k)$-SIIRV near $\p$ by looking at $\fourier{\h}$.

\begin{algorithm}[ht]
  \begin{algorithmic}[1]
    \Require Parameters $n$,$\eps$; the approximate Fourier coefficients $(\fourier{\h}(\xi))_{\xi\in S}$ modulo $M$, of a distribution $\p$ known to be effectively supported on $I$ and to have a Fourier transform effectively supported on $S$ of the form given in Step~\ref{algo:step:dft:computation} of~\cref{algo:ksiirv:tester}, with $\widetilde{\sigma}^2$
and $\widetilde{\mu}$, an approximation to $\E_{X \sim \p}[X]$ to within half a standard deviation.
	\State \label{step:cover:ksiirv} Compute $\mathcal{C}$, an $\frac{\eps}{5\sqrt{|S|}}$-cover in total variation distance of all $(n,k)$-SIIRVs.
	\For{each $\q \in \mathcal{C}$}
		\If{ the mean $\mu_\q$ and variance  $\sigma_\q$ of $\q$ satisfy $|\widetilde{\mu}-\mu_\q| \leq \widetilde{\sigma}$ and $2(\sigma_\q+1) \geq \widetilde{\sigma}+1 \geq (\sigma_\q+1)/2$} \label{step:cover:moment-test}
			\State Compute $\fourier{\q}(\xi)$ for $\xi \in S$.
			\If{ $\sum_{\xi \in S} |\fourier{\h}-\fourier{\q}|^2 \leq \frac{\eps^2}{5}$}
				\Return \accept
			\EndIf
		\EndIf
	\EndFor
	\State \Return \reject \Comment{we did not return \accept for any $\q \in \mathcal{C}$}
  \end{algorithmic}
  \caption{Algorithm \texttt{Project-k-SIIRV}}\label{algo:ksiirv-easy-projection}
\end{algorithm}

\begin{lemma}\label{lemma:project:ksiirv}
If Algorithm \texttt{Project-k-SIIRV} is given inputs that satisfy its assumptions and we have that $\sum_{\xi \in S} |\fourier{\h}-\fourier{\p}|^2 \leq (3\eps/25)^2$, $\dtv(\p,\h) \leq 6\eps/25$, and that if $\p\in\classksiirv[n]{k}$ then $\widetilde{\sigma}^2$ is a factor-$1.5$ approximation to $\var_{X \sim \p}[X]+1$, then it distinguishes between (i) $\p\in\classksiirv[n]{k}$ and (ii) $\dtv(\p,\classksiirv[n]{k})>\eps$. The algorithm runs in time $n\left( k/\eps \right)^{O(k\log(k/\epsilon))}$.
\end{lemma}

\subsubsection{The Case $k=2$}
For the important case of Poisson Binomial distributions, that is $(n,2)$-SIIRVs, we can dispense with using a cover at all. \cite{DKS15b} gives an algorithm that can properly learn Poisson binomial distributions in time $(1/\eps)^{O(\log \log 1/\eps)}$. The algorithm works by first learning the Fourier coefficients in $S$, which we have already computed here, and checks if one of many systems of polynomial inequalities has a solution: if the Fourier coefficients are close to those of a $(n,2)$-SIIRV, then there will be a solution to one of these systems. This allows us to test whether or not we are close to a $(n,2)$-SIIRV.

More precisely, we can handle this in two cases: the first, when the variance $s^2$ of $\p$ is relatively small, corresponding to $\widetilde{\sigma} \leq \alpha/\eps^2$ (for some absolute constant $\alpha >0$). In this case, we use the following lemma:

\begin{lemma}\label{lemma:project:pbd:smallvariance}
    Let $\p$ be a distribution with variance $O(1/\eps^2)$. Let $\widetilde \mu$ and $\widetilde{\sigma}^2$ be approximations to the mean $\mu$ and variance $s^2$ of $\p$ with $|\widetilde{\mu}-\mu| \leq \widetilde{\sigma}$ and $2(\sigma+1) \geq \widetilde{\sigma}+1 \geq (\sigma+1)/2$. 
Suppose that $\p$ is effectively supported on an interval $I$ and that its DFT modulo $M$ is effectively supported on $S$, the set of integers  $\xi \leq \ell\eqdef O(\log(1/\eps))$.
Let  $\fourier{\h}(\xi)$ be approximations to $\fourier{\p}(\xi)$ for all $\xi \in S$ with $\sum_{\xi \in S} |\fourier{\h}(\xi) - \fourier{\p}(\xi)|^2 \leq \frac{\eps^2}{16}$ .
 There is an algorithm that, given $n$,$\eps$,$\widetilde{\mu}$, $\widetilde{\sigma}$ and $\fourier{\h}(\xi)$, distinguishes between (i) $\p\in\classpbd[n]$ and (ii) $\dtv(\p,\classpbd[n])>\eps$, in time at most $(1/\eps)^{O(\log \log 1/\eps)}$.
 \end{lemma}
 
 If $\widetilde{\sigma} \geq \alpha/\eps^2$ (corresponding to a ``big variance'' $s^2 = \Omega(1/\eps^2)$), then we take an additional $O(|S|/\eps^2)$ samples from $\p$ and use them to learn a shifted binomial using algorithms \texttt{Learn-Poisson} and \texttt{Locate-Binomial} from \cite{DDS15-journal} that is within $O(\eps/\sqrt{|S|})$ total variation distance from $\p$. If these succeed,  we can check if its Fourier coefficients are close using the method in~\cref{algo:ksiirv-easy-projection} (\texttt{Project-k-SIIRV}). As we can compute the Fourier coefficients of a shifted binomial easily, this overall takes time $\poly(1/\eps)$.

%% file: general.tex
In this section, we abstract the ideas underlying the $(n,k)$-SIIRV from~\cref{sec:siirv:testing}, to provide a general testing framework. In more detail, our theorem (\cref{theo:testing:general}) has the following flavor: if $\mathcal{P}$ is a property of distributions such that every $\p\in\mathcal{P}$ has both (i) small effective support and (ii) sparse effective Fourier support, then one can test membership to $\mathcal{P}$ with $O(\sqrt{sM}/\eps^2+s/\eps^2)$ samples (where $M$ and $s$ are the bounds on the effective support and effective Fourier support, respectively). As a caveat, we do require that the sparse effective Fourier support $S$ be independent of $\p\in\mathcal{P}$, i.e., is a characteristic of the class $\mathcal{P}$ itself.

The high-level idea is then quite simple: the algorithm proceeds in three stages, namely the \emph{effective support test}, the \emph{Fourier effective support test}, and the \emph{projection step}. In the first, it takes some samples from $\p$ to identify what should be the effective support $I$ of $\p$, if $\p$ did have the property: and then checks that indeed $\abs{I}\leq M$ (as it should) and that $\p$ puts probability mass $1-O(\eps)$ on $I$.

In the second stage, it invokes the Fourier testing algorithm of~\cref{sec:fourier:support:testing} to verify that $\fourier{\p}$ indeed puts very little Fourier mass outside of $S$; and, having verified this, learns very accurately the set of Fourier coefficients of $\p$ on this set $S$, in $L_2$ distance.

At this point, either the algorithm has detected that $\p$ violates some required characteristic of the distributions in $\mathcal{P}$, in which case it has rejected already; or is guaranteed to have \emph{learned} a good approximation $\h$ of $\p$, by the Fourier learning performed in the second stage. It only remains to perform the third stage, which ``projects'' this good approximation $\h$ of $\p$ onto $\mathcal{P}$ to verify that $\h$ is close to some distribution $\p^\ast\in\mathcal{P}$ (as it should if indeed $\p\in\mathcal{P}$).

\begin{algorithm}[ht]
\algblockdefx[EFFECTIVE]{StartEffectiveSupport}{EndEffectiveSupport}{\textsf{Effective Support}}{}
\algblockdefx[FOURIER]{StartFourierSupport}{EndFourierSupport}{\textsf{Fourier Effective Support}}{}
\algblockdefx[PROJECT]{StartProjection}{EndProjection}{\textsf{Projection Step}}{}
  \begin{algorithmic}[1]
    \Require sample access to a distribution $\p\in\distribs{\N}$, parameter $\eps\in(0,1]$, $b\in(0,1]$, functions $S\colon(0,1]\to 2^{\N}$, $M\colon (0,1]\to\N$, $q_I\colon (0,1]\to\N$, and procedure $\textsc{Project}_{\mathcal{P}}$ as in~\cref{theo:testing:general}    
    \StartEffectiveSupport
      \State\label{algo:general:step:def:interval} Take $q_I(\eps)$ samples from $\p$ to identify a ``candidate set'' $I$.
            \Comment{Guaranteed to work w.p. $19/20$ if $\p\in\mathcal{P}$.}
      \State\label{algo:general:step:effectivesupport} Draw $O(1/\eps)$ samples from $\p$, to distinguish between $\p(I) \geq 1- \frac{\eps}{5}$ and $\p(I) < 1 - \frac{\eps}{4}$.
            \Comment{Correct w.p. $19/20$.}
      \If{$\abs{I} > M(\eps)$ or we detected that $\p(I) > \frac{\eps}{4}$}
        \State\label{algo:general:step:effectivesupportcheck} \Return \reject
      \EndIf
    \EndEffectiveSupport
    \StartFourierSupport
      \State\label{algo:general:step:fourier:support} Simulating sample access to $\p'\eqdef \p\bmod M(\eps)$, call~\cref{algo:ft:effective:support} on $\p'$ with parameters $M(\eps)$, $\frac{\eps}{5\sqrt{M(\eps)}}$, $b$, and $S(\eps)$.
      \If{\cref{algo:ft:effective:support} returned \reject}
        \State \Return \reject
      \EndIf
        \State Let $\fourier{\h}=(\fourier{\h}(\xi))_{\xi\in S(\eps)}$ denote the collection of Fourier coefficients it outputs, and $\h$ their inverse Fourier transform (modulo $M(\eps)$) \Comment{Do not actually compute $\h$ here.}
    \EndFourierSupport
    \StartProjection
      \State \label{algo:general:step:project} Call $\textsc{Project}_{\mathcal{P}}$ on parameters $\eps$ and $\h$, and \Return \accept if it does, \reject otherwise. 
    \EndProjection
  \end{algorithmic}
  \caption{Algorithm \texttt{Test-Fourier-Sparse-Class}}\label{algo:general:tester}
\end{algorithm}

\begin{theorem}[General Testing Statement]\label{theo:testing:general}
Assume $\mathcal{P}\subseteq \distribs{\N}$ is a property of distributions satisfying the following. There exist $S\colon(0,1]\to 2^{\N}$, $M\colon (0,1]\to\N$, and $q_I\colon (0,1]\to\N$ such that, for every $\eps\in(0,1]$,
\begin{enumerate}
  \item\label{condition:1} \textsf{Fourier sparsity:} for all $\p\in\mathcal{P}$, the Fourier transform (modulo $M(\eps)$) of $\p$ is concentrated on $S(\eps)$: namely,
  $
      \normtwo{\fourier{\p}\mathbf{1}_{\overline{S(\eps)}}}^2 \leq \frac{\eps^2}{100}
  $.
  \item\label{condition:2} \textsf{Support sparsity:} for all $\p\in\mathcal{P}$, there exists an interval $I(\p)\subseteq \N$ with $\abs{I(\p)} \leq M(\eps)$ such that (i) $\p$ is concentrated on $I(\p)$: namely,
  $
      \p(I(\p)) \geq 1-\frac{\eps}{5}
  $ and (ii) $I(\p)$ can be identified with probability at least $19/20$ from $q_I(\eps)$ samples from $\p$.
  \item\label{condition:3} \textsf{Projection:} there exists a procedure $\textsc{Project}_{\mathcal{P}}$ which, on input $\eps\in(0,1]$ and the explicit description of a distribution $\h\in\distribs{\N}$, runs in time $T(\eps)$; and outputs $\accept$ if $\dtv(\h,\mathcal{P})\leq \frac{2\eps}{5}$, and $\reject$ if $\dtv(\h,\mathcal{P}) > \frac{\eps}{2}$ (and can answer either otherwise).
  \item\label{condition:4} \textsf{(Optional) $L_2$-norm bound:} there exists $b\in(0,1]$ such that, for all $\p\in\mathcal{P}$, $\normtwo{\p}^2 \leq b$.
\end{enumerate}
Then, there exists a testing algorithm for $\mathcal{P}$, in the usual standard sense: it outputs either \accept or \reject, and satisfies the following.
    \begin{enumerate}
        \item if $\p \in \mathcal{P}$, then it outputs \accept with probability at least $3/5$;
        \item if $\dtv(\p,\mathcal{P}) > \eps$, then it outputs \reject with probability at least $3/5$.
    \end{enumerate}
The algorithm takes 
\[
    \bigO{ \frac{\sqrt{\abs{S(\eps)}M(\eps)}}{\eps^2} + \frac{\abs{S(\eps)}}{\eps^2} + q_I(\eps)  }
\] samples from $\p$ (if~\cref{condition:4} holds, one can replace the above bound by
$
    \bigO{ \frac{\sqrt{b}M(\eps)}{\eps^2} + \frac{\abs{S(\eps)}}{\eps^2} + q_I(\eps)  }
$); and runs in time $\bigO{m\abs{S} + T(\eps)}$, where $m$ is the sample complexity.

Moreover, whenever the algorithm outputs \accept, it also \emph{learns} $\p$; that is, it provides a hypothesis $\h$ such that $\dtv(\p,\h) \leq \eps$ with probability at least $3/5$.
\end{theorem}
We remark that the statement of~\cref{theo:testing:general} can be made slightly more general; specifically, one can allow the procedure $\textsc{Project}_{\mathcal{P}}$ to have sample access to $\p$ and err with small probability, and further provide it with the Fourier coefficients $\fourier{\h}$ learnt in the previous step.

\begin{proof}[Proof of~\cref{theo:testing:general}]
  For convenience, we hereafter write $S$ and $M$ instead of $S(\eps)$ and $M(\eps)$, respectively. Before establishing the theorem, which will be a generalization of (the second branch of)~\cref{algo:ksiirv:tester}, we note that it is sufficient to prove the version including~\cref{condition:4}. This is because, if no bound $b$ is provided, one can fall back to setting $b\eqdef\frac{\abs{S}+1}{M}$: indeed, for any $\p\in\mathcal{P}$,
  \begin{equation}\label{eq:default:bound:l2:norm}
      \normtwo{\p}^2 = \normtwo{\fourier{\p}}^2 = \normtwo{\fourier{\p}\mathbf{1}_{S}}^2+\normtwo{\fourier{\p}\mathbf{1}_{\bar{S}}}^2
      = \frac{1}{M}\sum_{\xi\in S}\dabs{\fourier{\p}(\xi)}^2 + \normtwo{\fourier{\p}\mathbf{1}_{\bar{S}}}^2
      \leq \frac{\abs{S}}{M} + \frac{\eps^2}{100M} = \frac{\abs{S}+\frac{\eps^2}{100}}{M}
  \end{equation}
  from~\cref{condition:1} and the fact that $\dabs{\fourier{\p}(\xi)}\leq 1$ for any $\xi\in[M]$. Then, we have $\sqrt{b}M \leq \sqrt{2\frac{\abs{S}}{M}}M = \sqrt{2\abs{S}M}$, concluding the remark. 
  
  The algorithm is given in~\cref{algo:general:tester}. Its sample complexity and running time are immediate from the assumptions on the input parameters, and its description; we thus focus on establishing its correctness.
  
  \begin{itemize}
      \item Completeness: suppose $\p\in\mathcal{P}$. Then, by definition of $q_I$ and $M$ (\cref{condition:2} of the theorem), we have that with probability at least $19/20$ the interval $I$ identified in Step~\ref{algo:general:step:def:interval} satisfies $\p(I)\geq 1-\frac{\eps}{5}$ and $\abs{I}\leq M$. In this case, also with probability at least $19/20$ the check in Step~\ref{algo:general:step:effectivesupport} succeeds, and the algorithm does not output \reject there.
      
      The call to~\cref{algo:ft:effective:support} in Step~\ref{algo:general:step:fourier:support} then, with probability at least 7/10, does not output \reject, but instead Fourier coefficients $\fourier{H}$ (supported on $S$) of some $\h$ such that $\h'=\h \bmod M$ satisfies $\normtwo{\h'-\p'} \leq \frac{6}{5}\cdot \frac{\eps}{5\sqrt{M}} = \frac{6\eps}{25\sqrt{M}}$ (this is because of the definition of $b$ and~\cref{condition:1}, which ensure the assumptions of~\cref{theo:ft:effective:support} are met). Thus $\normone{\h'-\p'} \leq \sqrt{M}\normtwo{\h'-\p'} \leq \frac{6\eps}{25}$. Since $\normtwo{\p-\p'} \leq 2\cdot\frac{\eps}{4}$ (as $\p(I)\geq 1-\frac{\eps}{4}$ and $\p'=\p\bmod M$), by the triangle inequality
      \[
        \dtv(\p,\h') = \frac{1}{2}\normone{\h'-\p'} \leq \frac{3\eps}{25}+\frac{\eps}{4} < \frac{2\eps}{5}
      \]
      and the algorithm returns \accept in Step~\ref{algo:general:step:project} (as promised by~\cref{condition:3}).
      
      Overall, by a union bound the algorithm is correct with probability at least $1-(\frac{1}{20}+\frac{1}{20}+\frac{3}{10}) \geq \frac{3}{5}$.
      
      \item Soundness: we proceed by contrapositive, and show that if the algorithm returns \accept with probability at least $3/5$ then $\dtv(\p,\mathcal{P}) \leq \eps$. 
We hereafter assume the guarantees of Steps~\ref{algo:general:step:def:interval},~\ref{algo:general:step:effectivesupport}, and~\ref{algo:general:step:fourier:support} hold, which by a union bound is the case with probability at least $1-(\frac{1}{20}+\frac{1}{20}+\frac{3}{10}) \geq \frac{3}{5}$.

Since the algorithm passed Step~\ref{algo:general:step:effectivesupportcheck}, we have $\p(I)\geq 1 - \frac{\eps}{4}$ and $\abs{I}\leq M$. Furthermore, as the algorithm did not output \reject on Step~\ref{algo:general:step:fourier:support}, by~\cref{theo:ft:effective:support} we know that the inverse Fourier transform (modulo $M$) $\h$ of the $S$-sparse collection of Fourier coefficients $\fourier{\h}$ returned satisfies, for $\h'\eqdef \h \bmod M$,
  \[
       \normtwo{\h'-\p'} \leq \frac{6\eps}{25\sqrt{M}}
  \]
  which by Cauchy--Schwarz implies that $\normone{\h-\p'}\leq \frac{6\eps}{25}$, or equivalently $\dtv(\h,\p')\leq \frac{3\eps}{25}$.
  
  Finally, since the algorithm outputted \accept in Step~\ref{algo:general:step:project}, there exists $\p^\ast\in\mathcal{P}$ (supported on $[M]$) such that $\dtv(\h,\p^\ast)\leq \frac{\eps}{2}$, and by the triangle inequality
  \[
      \dtv(\p,\p^\ast) \leq \dtv(\p,\p') + \dtv(\h,\p') + \dtv(\h,\p^\ast) \leq \frac{\eps}{4} + \frac{3\eps}{25} + \frac{\eps}{2} \leq \eps
  \]
  and thus $\dtv(\p,\mathcal{P}) \leq \dtv(\p,\p^\ast) \leq \eps$.
      
  \end{itemize}
\end{proof}

%% file: pmdtesting.tex
In this section, we generalize our Fourier testing approach to higher dimensions, and leverage it to design a testing algorithm for the class of Poisson Multinomial distributions -- thus establishing~\cref{theo:testing:pmd} (restated below).

\begin{theorem}[Testing PMDs]
    Given parameters $k,n\in\mathbb{N}$, $\eps\in(0,1]$, and sample access to a distribution $\p$ over $\N$, there exists an algorithm (\cref{algo:pmd:tester}) which outputs either \accept or \reject, and satisfies the following.
    \begin{enumerate}
        \item if $\p \in \classpmd[n]{k}$, then it outputs \accept with probability at least $3/5$;
        \item if $\dtv(\p,\classpmd[n]{k}) > \eps$, then it outputs \reject with probability at least $3/5$.
    \end{enumerate}
    Moreover, the algorithm takes $\bigO{\frac{n^{(k-1)/4} k^{2k} \log(k/\eps)^k}{\eps^2}}$ samples from $\p$, and runs in time $n^{O(k^3)} \cdot (1/\eps)^{O(k^3\frac{\log(k/\eps)}{\log\log(k/\eps)})^{k-1}}$ or alternatively in time $n^{O(k)} \cdot  2^{O(k^{5k} \log(1/\eps)^{k+2}}$.
\end{theorem}
The reason for the two different running times is that, for the projection step, one can use either the cover given by \cite{DKS15b} or that given by \cite{DDKT16}, which yield the two statements.
In contrast to~\cref{sec:siirv:testing} and~\cref{sec:general:testing}, for PMDs we will have to use a \emph{multidimensional} Fourier transform, which is a little more complicated -- and we define next.

Let $M \in \Z^{k \times k}$ be an integer $k \times k$ matrix.
We consider the integer lattice
$L  = L(M) =  M \Z^k \eqdef \{ p \in \Z^k \mid p = M q,  q \in \Z^k \}$, and its dual lattice
 $L^{\ast} = L^{\ast}(M)   \eqdef \setOfSuchThat{ \xi \in \R^k }{ \xi \cdot x \in \Z \textrm{ for all } x \in L }.$
 Note that  $L^{\ast} = (M^T)^{-1} \Z^k,$ and that $L^{\ast}$ is not necessarily integral.
The quotient $ \Z^k \slash L$ is the set of equivalence classes of points in $\Z^k$ such that two points $x, y \in \Z^k$
are in the same equivalence class if, and only if, $x - y \in L$.
Similarly, the quotient $L^{\ast} \slash \Z^k$ is the set of equivalence
classes of points in $L^{\ast}$ such that any two points $x, y \in L^{\ast}$ are in the same equivalence class if, and only if, $x -y \in \Z^k$.

The \emph{Discrete Fourier Transform (DFT) modulo $M$}, $M \in \Z^{k \times k}$, of a function
$F\colon \Z^k \to \C$ is  the function $\widehat{F}_M\colon L^{\ast} \slash \Z^k  \to \C$
defined as $\widehat{F}_M(\xi)\eqdef\sum_{x \in  \Z^k} e(\xi \cdot x) F(x).$ 
(We will omit the subscript $M$ when it is clear from the context.)
Similarly, for the case that $F$ is a probability mass function, we can equivalently write
$\widehat{F}(\xi)= \E_{X \sim F} \left[ e(\xi \cdot X) \right].$ The \emph{inverse DFT} of a function $\widehat{G}\colon L^{\ast} \slash \Z^k  \to \C$
is the function $G\colon A \to \C$ defined on a \emph{fundamental domain} $A$ of $L(M)$ as follows:
$G(x) = \frac{1}{|\det(M)|} \sum_{\xi \in L^{\ast} \slash \Z^k} \widehat{G}(x) e(- \xi \cdot x).$
Note that these operations are inverse of each other,
namely for any function $F\colon A \to \C$, the inverse DFT of $\widehat{F}$ is identified with  $F$.

With this in hand, \cref{algo:ft:effective:support} easily generalizes to high dimension:
\begin{algorithm}
  \begin{algorithmic}[1]
    \Require parameters, a $k \times k$ matrix $M$, $b,\eps\in(0,1]$; a fundamental domain $A$ of $L(M)$; sample access to distribution $\q$ over $A$
    \State\label{algo:ft:pmd:step:choosemprime} Set $m\gets \clg{C(\frac{\sqrt{b}}{\eps^2}+ \sqrt{\det(M)})}$ \Comment{$C>0$ is an absolute constant; $C=2000$ works.}
    \State Draw $m'\gets \Poi(m)$; if $m'>2m$, \Return \reject
    \State\label{algo:ft:pmd:step:empr} Draw $m'$ samples from $\q$, and let $\q'$ be the corresponding empirical distribution over $[M]$
    \State\label{algo:ft:pmd:step:norm} Compute $\normtwo{\q'}^2$, $\fourier{\q'}(\xi)$ for every $\xi\in S$, and $\normtwo{\fourier{\q'}\mathbf{1}_S}^2$ \Comment{Takes time $\bigO{m\abs{S}}$}
    \If{ $m'^2\normtwo{\q'}^2 - m' > \frac{3}{2}bm^2$ }\label{algo:ft:pmd:step:norm:check} \Return \reject
    \ElsIf{ $\normtwo{\q'}^2 - \normtwo{\fourier{\q'}\mathbf{1}_S}^2 \geq 3\eps^2+\frac{1}{m'}$ } \Return \reject
    \Else
      \State \Return $(\fourier{\q'}(\xi))_{\xi\in S}$
    \EndIf
  \end{algorithmic}
  \caption{Testing the Fourier Transform Effective Support in high dimension}\label{algo:ft:pmd:effective:support:high-dim}
\end{algorithm}

Crucially, we observe that the proof of~\cref{theo:ft:effective:support} nowhere requires that $[M]$ be a set of $M$ consecutive integers, but only that it is a fundamental domain of the lattice used in the DFT. Consequently,~\cref{theo:ft:effective:support} also applies in this high dimensional setting, with appropriate notation. Note that the size of any fundamental domain is $\det(M)$ which appears in place of $M$ in the sample complexity.

\begin{algorithm}[ht]
  \begin{algorithmic}[1]
    \Require sample access to a distribution $\p\in\distribs{\N^k}$, parameters $n,k\geq 1$ and $\eps\in(0,1]$
    
    \State\Comment{ Let $C,C',C''$ be sufficiently large universal constants }
    
    \State\label{algo:step:estimates:mu:sigma:pmd} Draw {$m_0 = O(k^4)$} samples from $X$, and 
let $\wh{\mu}$ be the sample mean and $\wh{\Sigma}$ the sample covariance matrix.
	
	\State Compute an approximate spectral decomposition of $\wh{\Sigma}$, i.e., 
an orthonormal eigenbasis $v_i$ with corresponding eigenvalues $\lambda_i$, $i \in [k].$

	\State Set $M \in \Z^{k \times k}$  to be the matrix whose $i^{th}$ column 
is the closest integer point to the vector $C \left(\sqrt{k \log(k/\eps)\lambda_i+k^2\log^2(k/\eps)}\right)v_i.$
	
	\State Set $I \gets \Z^k\cap (\wh{\mu} + M \cdot (-1/2,1/2]^k)$

	\State\label{algo:step:effectivesupport:pmd} Draw $O(1/\eps)$ samples from $\p$, and \Return \reject if any falls outside of $I$

	\State\label{algo:step:dft:computation:pmd} Let $S \subseteq (\R/\Z)^k$ to be the set of points $\xi = (\xi_1, \ldots, \xi_k)$ 
of the form $\xi = (M^T)^{-1} \cdot v {+\Z^k},$ for some $v\in \Z^k$ with $\|v\|_2 \leq C^2 k^2 \log(k/\eps).$

	\State Define $\p\bmod M$ to be the distribution obtained by sampling $X$ from $\p$ and if it lies outside in $I$, returning $X$, else returning $X + M b$ for the uniwue $b \in \Z^k$ such that $X + M b \in I$.

	\State\label{algo:step:fourier:support:pmd} Simulating sample access to $\p'\eqdef \p\bmod M$, call~\cref{algo:ft:pmd:effective:support:high-dim} on $\p'$ with parameters $M$, $\frac{\eps}{5\sqrt{\det(M)}}$, $b=\frac{|S|+1}{\det(M)}$, and $S$. If it outputs \reject, then \Return \reject; otherwise, let $\fourier{\h}=(\fourier{\h}(\xi))_{\xi\in S}$ denote the collection of Fourier coefficients it outputs, and $\h$ their inverse Fourier transform (modulo $M$) onto $I$. \Comment{Do not actually compute $\h$} 

    \State \label{algo:step:cover:pmd} Compute a proper $\eps/6\sqrt{|S|}$-cover $\mathcal{C}$ of all PMDs using the algorithm from \cite{DKS15c}.
	
	\For{each $\q \in \mathcal{C}$}
		\If{ the mean $\mu_\q$ and covariance matrix $\Sigma_\q$ satisfy $(\wh{\mu}-\mu_\q)^T(\Sigma+I)^{-1}(\wh{\mu}-\mu_\q) \leq 1$ and $2(\Sigma_\q+I) \geq \wh{\Sigma}+I \geq (\Sigma_\q+I)/2.$}
			\State Compute $\fourier{\q}(\xi)$ for $\xi \in S$.
			\If{ $\sum_{\xi \in S} |\fourier{\h}-\fourier{\q}|^2 \leq \eps^2/16$}
				\Return \accept
			\EndIf
		\EndIf
	\EndFor
	\State \Return \reject if we do not \accept for any $\q \in \mathcal{C}$.
  \end{algorithmic}
  \caption{Algorithm \texttt{Test-PMD}}\label{algo:pmd:tester}
\end{algorithm}

The proof of correctness of~\cref{algo:pmd:tester} is very similar to that of~\cref{algo:ksiirv:tester}, except that we need results from the proof of correctness of the PMD Fourier learning algorithm of~\cite{DKS15c}; we will only sketch these ingredients here. That $I$ is an effective support of a PMD whose mean and covariance matrix we have estimated to within appropriate error with high probability follows from Lemmas 3.3--3.6 of \cite{DKS15c}, the last of which gives that the probability mass outside of $I$ is at most $\eps/10$, smaller than that claimed for $I$ in the $(n,k)$-SIIRV algorithm. Lemma 3.3 gives, if $\p$ is a PMD, that the mean and covariance satisfy $(\wh{\mu}-\mu)^T(\Sigma+I)^{-1}(\wh{\mu}-\mu) = O(1)$ and $2(\Sigma_\q+I) \geq \wh{\Sigma}+I \geq (\Sigma_q+I)/2.$ Again, with more samples, we can strengthen this to $(\wh{\mu}-\mu)^T(\Sigma+I)^{-1}(\wh{\mu}-\mu) = \frac{1}{2}$ and $(3/2)(\Sigma+I) \geq \wh{\Sigma}+I \geq (\Sigma+I)/(3/2)$ with $O(k^4)$ samples.

\noindent The effective support of the Fourier transform of a PMD is given by the following proposition:
\begin{proposition}[Proposition 2.4 of \cite{DKS15c}] \label{prop:ft-effective-support}
Let $S$ be as in the algorithm. With probability at least $99/100$, the Fourier coefficients of $\p$ outside $S$ satisfy
$\sum_{\xi \in (L^{\ast}/\Z^k)  \setminus S }|\wh{\p}(\xi)| < \eps/10.$

This holds not just for $\p$, but any $(n,k)$-PMD $\q$ whose mean $\mu_\q$ and covariance matrix $\Sigma_\q$ satisfy $(\wh{\mu}-\mu_\q)^T(\Sigma+I)^{-1}(\wh{\mu}-\mu) = O(1)$ and $2(\Sigma_\q+I) \geq \wh{\Sigma}+I \geq (\Sigma_\q+I)/2.$
\end{proposition}

We need to show that this $L_1$ bound is stronger than the $L_2$ bound we need. Since every individual $\xi \notin S$ has $|\wh{\p}(\xi)| < \eps/10$, we have
$$\sum_{\xi \in (L^{\ast}/\Z^k)  \setminus S }|\wh{\p}(\xi)|^2 \leq \sum_{\xi \in (L^{\ast}/\Z^k)  \setminus S } \eps |\wh{\p}(\xi)|/10 \leq \eps^2/100$$
and so $S$ is an effective support of the DFT modulo $M$.

To show that the value of $b$ is indeed a bound on $\normtwo{\p}^2$, we can use (\ref{eq:default:bound:l2:norm}), yielding that $\normtwo{\p}^2 \leq (|S|+1)/\det(M) = b $, where $\det(M)$ here is indeed the size of $I$.

The proof of correctness of the algorithm and the projection step is now very similar to the $(n,k)$-SIIRV case. We need to get bounds on the sample and time complexity.
We can bound the size of $S$ using
\begin{align*}
|S| & \leq \abs{\setOfSuchThat{ v\in \Z^k }{ \|v\|_2 \leq  C^2 k^2 \log(k/\eps) }} 
\leq \abs{\setOfSuchThat{ v\in \Z^k }{ \|v\|_\infty \leq  C^2 k^2 \log(k/\eps) }} \\
&= \left(1+2\lfloor C^2 k^2 \log(k/\eps)\rfloor \right)^k  
 = O(k^2\log(k/\eps))^k
\end{align*}
We can bound $\det(M)$ in terms of the $L_2$ norms of its columns using Hadamard's inequality
\[
  \det(M) \leq \prod_{i=1}^k \normtwo{M_i} \leq \prod_{i=1}^k \left( C \left(\sqrt{k \log(k/\eps)\lambda_i+k^2\log^2(k/\eps)}\right) + \sqrt{k} \right)
\]
recalling that $\lambda_i$ are the eigenvalues of $\wh{\Sigma}$ which satisfies $2(\Sigma_\q+I) \geq \wh{\Sigma}+I$.
 We need a bound on $\|\Sigma\|_2$. Each individual summand $k$-CRV (categorical random variable) is supported on unit vectors, the distance between any two of which is $\sqrt{2}$. Therefore we have that $\|\Sigma\|_2 \leq 2n$. Then $\lambda_i \leq 4n+1$ for every $1\leq i\leq k$; moreover, since the $k$ coordinates must sum to $n$, $\wh{\Sigma}$ has rank at most $k-1$ and so at least one of the $\lambda_i$'s is zero. Combining these observations, we obtain
\[
    \det(M) \leq \sqrt{k^2\log^2\frac{k}{\eps}}\cdot \left(C^2 k(4n+2) \log\frac{k}{\eps} + k^2\log^2\frac{k}{\eps}\right)^{\frac{k-1}{2}} = k\log\frac{k}{\eps}\cdot \bigO{nk^2 \log\frac{k}{\eps}}^{\frac{k-1}{2}} \; .
\]
With high constant probability, the number of samples we need is then
 \begin{align*}
 & \bigO{ \frac{\sqrt{\abs{S}\det{M}}}{\eps^2} + \frac{\abs{S}}{\eps^2} + q_I(\eps)  } = 
  \frac{1}{\eps^2}\sqrt{k\log\frac{k}{\eps}}\cdot \bigO{nk^2 \log\frac{k}{\eps}}^{\frac{k-1}{4}} + \frac{O(k^2\log(k/\eps))^{k}}{\eps^2} + O(k^4) \\
 &= O(n^{(k-1)/4} k^{2k} \log(k/\eps)^k/\eps^2)
 \end{align*}
 The time complexity of the algorithm is dominated by the projection step. By Proposition 4.9 and Corollary 4.12 of \cite{DKS15c}, we can produce a proper
 $\eps$-cover of $\classpmd[n]{k}$ of size $n^{O(k^3)} \cdot (1/\eps)^{O(k^3\frac{\log(k/\eps)}{\log\log(k/\eps)})^{k-1}}$ in time also $n^{O(k^3)} \cdot (1/\eps)^{O(k^3\frac{\log(k/\eps)}{\log\log(k/\eps)})^{k-1}}$. Note that producing an $(\eps/6\sqrt{|S|})$-cover, as $=\eps/O(k^2\log(k/\eps))^{k/2}$, takes time $n^{O(k^3)} \cdot (1/\eps)^{O(k^3\frac{\log(k/\eps)}{\log\log(k/\eps)})^{k-1}}$ (which is also the size of the resulting cover). Hence the running time of the algorithm is at most $n^{O(k^3)} \cdot (1/\eps)^{O(k^3\frac{\log(k/\eps)}{\log\log(k/\eps)})^{k-1}}$.

 Alternatively, \cite{DDKT16} gives an $\eps$-cover of size $n^{O(k)} \cdot \min{2^{\poly(k/\eps)}, 2^{O(k^{5k} \log(1/\eps)^{k+2}}}$ that can also be constructed in polynomial time. By using this result, one needs to take time $n|S|\poly(\log(1/\eps))$ to compute the Fourier coefficients. Applying this to get an $\eps/O(k^2\log(k/\eps))^{k/2}$-cover means that unfortunately we are always doubly exponential in $k$. In this case, the running time of the algorithm is $n^{O(k)} \cdot  2^{O(k^{5k} \log(1/\eps)^{k+2}}$.

%% file: logconcaves.tex
\begin{theorem}[Testing Log-Concavity]
Given parameters $n\in\N$, $\eps\in(0,1]$, and sample access to a distribution $\p$ over $\Z$, there exists an algorithm which outputs either \accept or \reject, and satisfies the following.
    \begin{enumerate}
        \item if $\p \in \classlogconcave_n$, then it outputs \accept with probability at least $3/5$;
        \item if $\dtv(\p,\classlogconcave_n) > \eps$, then it outputs \reject with probability at least $3/5$.
    \end{enumerate}
where $\classlogconcave_n$ denotes the class of (discrete) log-concave distributions over $\{0,\dots,n-1\}$. Moreover, the algorithm takes $O(\sqrt{n}/\eps^2 + \log(1/\eps)/\eps^{5/2})$ samples from $\p$; and runs in time $O(\sqrt{n} \cdot \poly(1/\eps))$.
\end{theorem}

We will sketch the proof and algorithm here. We first remark that the Maximum Likelihood Estimator (MLE) for log-concave distributions can be formulated as a convex program \cite{DR11}, 
which can be solved in sample polynomial time. One advantage of the MLE for log-concave distributions is that it properly learns log-concave distributions (over support size $M$) 
to within Hellinger distance $\eps$ using $O(\log (M/\eps)/\eps^{5/2})$ samples\footnote{We note that a similar, slightly stronger result is already known for \emph{continuous} log-concave distributions, which can be learned to Hellinger distance $\eps$ from only $O(\eps^{-5/2})$ samples~\cite{KS:16}. The proof of this result, however, does not seem to generalize to discrete log-concave distributions, which is our focus here; thus, we establish in~\cref{appendix:log:concave} the learning result we require, namely an upper bound on the sample complexity of the MLE estimator for learning the class of log-concave distributions over $\{0,\dots,M-1\}$ in Hellinger distance (\cref{theo:mle:logconcave}).}. Note that the squared Hellinger distance satisfies:
\[
\hellinger{\p}{\q}^2 = \sum_x (\sqrt{\p(x)}-\sqrt{\q(x)})^2 = \sum_x \frac{(\p(x)-\q(x))^2}{(\sqrt{\p} +\sqrt{\q})^2} \geq \frac{\normtwo{\p-\q}}{2\max\{\p(x),\q(x)\}} \;.
\]
Further, it is known that a log-concave distribution with variance $\sigma^2$ is effectively supported in an interval of length $M=O(\log(1/\eps) \sigma)$ centered at the mean, 
and that its maximum probability is $O(1/\sigma)$ (See~\cref{fact:log-concave-standard}). Thus, by learning a log-concave distribution properly to within $\eps/\log(1/\eps)$ Hellinger distance, one also learns it to within $\frac{\eps}{\sqrt{M}}$ $L_2$-distance.

A log-concave distribution $\p$ has $L_2$ norm bounded by $\normtwo{\p}^2 \leq \max_x \p(x) \leq O(1/\sigma)$. 
It is easy to show using concentration bounds(\cref{fact:log-concave-standard}) that $\p \bmod M$ also has $L_2$ norm $O(1/\sqrt{\sigma})$. 
We will prove in~\cref{prop:log-concave-sparse-FT} that its DFT modulo $M$ is effectively 
supported on a known set $S$ of size $|S|=O(\log(1/\eps)^2/\eps^2)$.
\new{
\begin{algorithm}[ht]
  \begin{algorithmic}[1]
    \Require sample access to a distribution $\p\in\distribs{\N}$, parameter $\eps\in(0,1]$
    
    \State\Comment{ Let $C,C',C''$ be sufficiently large universal constants }
    
    \State Draw $O(1)$ samples from $\p$ and compute their mean $\widetilde{\mu}$ and let $\widetilde{\sigma}$ be $1$ plus their standard deviation.
	\State Set $M \gets 1+2 \clg{ C \widetilde{\sigma}\ln(1/\eps) }$, and let $I \gets [\flr{ \widetilde{\mu} }-\frac{M-1}{2},\flr{ \widetilde{\mu} }+\frac{M-1}{2}]$
	\State Draw $O(1/\eps^2)$ samples from $\p$, to distinguish between $\p(I) \leq 1-\frac{\eps^2}{4}$ and $\p(I) > 1-\frac{\eps^2}{5}$.
	If the former is detected, \Return \reject
	\State Draw $O(\log( M/\eps)/\eps^{5/2})$ samples from $\p$ and let $T$ be the subset of these samples in $I$. Compute the MLE $H$ over all discrete log-concave distributions for $T$ using a convex program.
	
	\State Compute the standard deviation $\sigma_{\h}$ of $\h$. If $1+\sigma_{\h} \leq \widetilde{\sigma}/2$ or $\sigma_{\h} \geq 2 \widetilde{\sigma}$, then \Return \reject.
	
	\State Set $S \gets \setOfSuchThat{ \xi \in [M-1] }{ |\xi| \leq C' \log(1/\eps)^2/\eps^2}$
	\State Let $m=C''/(\eps^2 \sqrt{\widetilde{\sigma}})$ and  draw $m'$ from $\Poi(m)$. Take $m'$ samples from $\p$ and let $\q'$ be their empirical distibution.
	
	\State Compute $\fourier{\q'}(\xi)$ and $\fourier{\h}(\xi)$ for every $\xi \in S$.
	\If {$\normtwo{\q'}^2-\normtwo{\fourier{\q'}\mathbf{1}_S}^2/M + \normtwo{\fourier{\q'}\mathbf{1}_S-\fourier{\h}\mathbf{1}_S}^2/M > 3m^2 \eps^2$}
		 \State \Return \reject
	\Else
		\State \Return \accept
	\EndIf
	
  \end{algorithmic}
  \caption{Algorithm \texttt{Test-log-concave}}\label{algo:log-con:tester}
\end{algorithm}
}

Thus our algorithm (Algorithm \ref{algo:log-con:tester}) will work as follows: First we estimate the mean and variance under the assumption of log-concavity. 
We construct an interval $I$ of length $M=O(\log(1/\eps) \sigma)$ which would be containing the effective support if we were log-concave; 
and reject if it is not the case, i.e., too much probability mass falls outside $I$. Then we properly learn $\p$ to within $\eps/\log(1/\eps)$ Hellinger distance 
using the MLE of $O(\log (M/\eps)/\eps^{5/2})$ samples,\footnote{Note that we here invoke the MLE estimator not on the full domain, but on the effective support, which contains at least $1-O(\eps^2)$ probability mass. This conditioning overall does not affect the sample complexity nor the distances, as it can only cause $O(\eps^2)$ error in total variation (and thus $O(\eps)$ in Hellinger distance).}\ giving a hypothesis $\h$. At this point, we reject if our estimates for the variance 
is far from that of $\h$. Then we run an $L_2$ identity tester between $\p$ and $\h$, i.e., test whether the empirical distribution $\q$ of $O(M/\sigma\eps^2)$ samples is far in $L_2$ from $\h$ \new{using Proposition \ref{prop:l2:identity:tester}} .
To do this efficiently, we compute $\normtwo{\q'}^2-\normtwo{\fourier{\q}\mathbf{1}_S}^2/M + \normtwo{\fourier{\q'}\mathbf{1}_S-\fourier{\h}\mathbf{1}_S}^2/M$ \new{which is close to $\normtwo{\q' - \h}^2$ }
 since $\fourier{\h}$ is effectively supported on $S$ as it is a log-concave distribution whose stamdard deviation is at least half of our estimate. 

To do this in time $O(\sqrt{n} \cdot \poly(1/\eps))$, we need to compute the Fourier coefficients efficiently. The MLE for log-concave distributions 
is a piecewise exponential distribution with a number of pieces at most the number of samples~\cite{DR11}, 
which is $O(\log (M/\eps)/\eps^{5/2})$ in this case. Using the expression for \new{the sum of a geometric series}
gives a simple closed-form expression for $\fourier{\h}(\xi)$ that we can compute in time $O(\log (M/\eps)/\eps^{5/2})$.

\begin{proposition} \label{prop:log-concave-sparse-FT} 
Let $\p$ be a discrete log-concave distribution with variance $\sigma^2$ and  $M = O(\log(1/\eps) \sigma)$ be the size of its effective support. 
Then its Discrete Fourier transform is effectively supported on a known set $S$ of size $|S|=O(\log(1/\eps)^2/\eps^2)$. 
\end{proposition}
\begin{proof}
First we show that for any unimodal distribution, we can relate the maximum probability to the size of the effective support.

\begin{lemma} 
Let $\p$ be a unimodal distribution supported on $\Z$ such that the probability of the mode is $\p_{\max}$. 
Then the DFT modulo $M$ of $\p$ at $\xi \in [-M/2,M/2)$ has $\fourier{\p}(\xi)=O(\p_{\max} M/|\xi|)$.
\end{lemma}
\begin{proof}
Let $m$ be the mode of $\p$. Then we have 
\[
\fourier{\p}(\xi)=\sum_{j=-\infty}^{m-1} \p(j) \exp\!\left(-2\pi i \frac{\xi j}{M}\right) + \sum_{j=m}^\infty \p(j) \exp\!\left(-2\pi i \frac{\xi j}{M}\right)\,.
\]
We will apply summation by parts to these two series. 
Let $g(x) = \sum_{j=m+1}^x \exp(-2\pi i \xi j/M)$ and $g(m)=0$. 
By a standard result on geometric series, we have $g(x)= -\frac{\exp(-2\pi i \xi (x+1)/M) - \exp(-2\pi i \xi (m+1)/M)}{1- \exp(-2\pi i \xi/M)}$. 
\begin{claim} 
$|g(x)| = O(M/\xi)$ for all integers $x \geq m$. 
\end{claim}
\begin{proof}
The modulus of the numerator $|\exp(-2\pi i \xi (x+1)/M) - \exp(-2\pi i \xi (m+1)/M)|$ is at most $2$. 
We thus only need to find a lower bound for $|1- \exp(-2\pi i \xi/M|$.
\begin{align*}
|1- \exp(-2\pi i \xi/M)|^2 & = (1-\cos(2\pi\xi/M))^2 + \sin(2\pi\xi/M)^2 
 = 2 - 2 \cos(2\pi\xi/M) 
 = \Omega((\xi/M)^2) \;,
\end{align*}
and so $|g(x)| \leq 2/\sqrt{\Omega((\xi/M)^2)}=O(M/|\xi|)$.
\end{proof}
\noindent Now consider the following, for any $n > m$:
\[
  \sum_{j=m+1}^n \p(j) (g(j)-g(j-1)) + \sum_{j=m+1}^n g(j) (\p(j+1)-\p(j))  = \p(n+1)g(n)-\p(m+1)g(m) \; .
\]
Now $g(m)=0$ and $\p(n+1) \rightarrow 0$ as $n \rightarrow \infty$ while $g(n+1)$ is bounded for all $n$. 
Hence, the RHS tends to $0$ as $n \rightarrow \infty$ and we have:
\begin{align*} 
|\sum_{j=m+1}^\infty \p(j) \exp(-2\pi i \xi j/M)| 
&= |\sum_{j=m+1}^\infty \p(j) (g(j)-g(j-1))| 
= |\sum_{j=m+1}^\infty g(j) (\p(j+1)-\p(j)) | \\
& \leq O(M/\xi) \cdot \sum_{j=m+1}^\infty \left( \p(j)- \p(j+1) \right)
 = O(\p_{\max}M/\xi) \;.
\end{align*}
Similarly, we can show that $\sum_{j=-\infty}^{m-1} \p(j) \exp(-2\pi i \xi j/M)= O(\p_{\max}M/\xi)$ since $\p$ is monotone there as well.
\end{proof}
\noindent Then we can get a bound on the size of the effective support:
\begin{lemma} \label{lem:S-unimodal}
Let $\p$ be a unimodal distribution supported on $\Z$ such that the probability of the mode is $\p_{\max}$ and let $\eps \leq 1/M$. 
Then the DFT modulo $M$ of $\p$ has $\sum_{|\xi| > \ell} |\fourier{\p}|^2 \leq \eps^2/100$, where $\ell= \Theta(\p_{\max}^2M^2/\eps^2)$.
\end{lemma}
\begin{proof}
\begin{align*}
\sum_{|\xi| > \ell} |\fourier{P}|^2 & \leq 2\sum_{\xi=\ell+1}^{M/2} O(\p_{\max}M/\xi)^2 
 \leq O(\p_{\max}M)^2 \sum_{\xi=\ell+1}^\infty 1/\xi^2 
 \leq O(\p_{\max}^2M^2/\ell) \leq \frac{\eps^2}{100} \;. 
\end{align*}
\end{proof}
For log-concave distributions, we can relate $\p_{\max}$ and $M$ as follows,
\begin{fact} \label{fact:log-concave-standard} 
Let $\p$ be a discrete log-concave distribution with mean $\mu$ and variance $\sigma^2$. 
Then 
\begin{itemize}
\item $\p$ is unimodal;
\item its probability mass function satisfies $\p(x)=\exp(-O((x-\mu)/\sigma))/\sigma$; and
\item $\Pr[|X-\mu| \geq \Omega(\sigma \log(1/\eps))] \leq \eps$.
\end{itemize}
\end{fact}
Since $\p_{\max}=O(1/\sigma)$, we can take $M=O(\sigma \log(1/\eps)))=O(\log(1/\eps)/\p_{\max})$. 
Substituting this into Lemma \ref{lem:S-unimodal} completes the proof of the proposition.
\end{proof}

%% file: lowerbounds.tex
In this section, we obtain a lower bound to complement our upper bound for testing Poisson Multinomial Distributions. Namely, we prove the following:
\begin{theorem}\label{theo:lb:pmd}
  There exists an absolute constant $c\in(0,1)$ such that the following holds. For any $k\leq n^c$, any testing algorithm for the class of $\classpmd[n]{k}$ must have sample complexity
  $
    \bigOmega{ \left(\frac{4\pi}{k}\right)^{{k}/{4}}\frac{n^{{(k-1)}/{4}}}{\eps^2} }
  $.
\end{theorem}
The proof will rely on the lower bound framework of~\cite{CDGR:16}, reducing testing $\classpmd[n]{k}$ to testing identity to some suitable hard distribution $\p^\ast\in\classpmd[n]{k}$. To do so, we need to (a) choose a convenient $\p^\ast\in\classpmd[n]{k}$; (b) prove that testing identity to $\p^\ast$ requires that many samples (we shall do so by invoking the~\cite{VV14} instance-by-instance lower bound method); (c) provide an agnostic learning algorithm for $\classpmd[n]{k}$ with small enough sample complexity, for the reduction to go through. Invoking~\cite[Theorem 18]{CDGR:16} with these ingredients will then conclude the argument.
\begin{proof}[Proof of~\cref{theo:lb:pmd}]
In what follows, we choose our ``hard instance'' $\p^\ast\in\classpmd[n]{k}$ to be the PMD obtained by summing $n$ i.i.d. random variables, all uniformly distributed on $\{e_1,\dots,e_k\}$. This takes care of point (a) above.

To show (b), we will rely on a result of Valiant and Valiant, which showed in~\cite{VV14} that testing identity to any discrete distribution $\p$ required $\bigOmega{\norm{\p^{-\max}_{-\eps}}_{2/3}/\eps^2}$ samples, where $\p^{-\max}_{-\eps}$ is the vector obtained by zeroing out the largest entry of $\p$, as well as a cumulative $\eps$ mass of the smallest entries. Since $\norm{\p^{-\max}_{-\eps}}_{2/3}$ is rather cumbersome to analyze, we shall instead use a slightly looser bound, considering $\normtwo{\p}$ as a proxy.
\begin{fact}\label{fact:23:2:holder}
For any discrete distribution $\p$, we have $\norm{\p}_{2/3} \geq \frac{1}{\normtwo{\p}}$. More generally, for any vector $x$ we have $\norm{x}_{2/3} \geq \frac{\normone{x}^2}{\normtwo{x}}$.
\end{fact}
\begin{proof}
It is sufficient to prove the second statement, which implies the first. This is in turn a straightforward application of H\"older's inequality, with parameters $(4,\frac{4}{3})$:
$
    \normone{x} = \sum_{i} \abs{x}_i^{1/2}\abs{x}_i^{1/2} \leq \left( \sum_{i} \abs{x}_i^{2}\right) ^{1/4} \left( \abs{x}_i^{2/3}\right)^{3/4}
$. Squaring both sides yields the claim.
\end{proof}
\begin{fact}
For our distribution $\p^\ast$, we have $\normtwo{\p^\ast} = \bigTheta{ \frac{k^{k/4}}{ (4\pi n)^{(k-1)/4}} }$.
\end{fact}
\begin{proof}
It is not hard to see that, from any $\mathbf{n}=(n_1,\dots,n_k)\in\N^k$ such that $\sum_{i=1}^k n_i = n$, $\p^\ast( \mathbf{n} ) = \frac{1}{k^n} \binom{n}{n_1,\dots,n_k}$ (where $\binom{n}{n_1,\dots,n_k}$ denotes the multinomial coefficient). From there, we have
\[
    \normtwo{\p^\ast}^2 = \frac{1}{k^{2n}}\sum_{n_1+\dots+n_k=n} \binom{n}{n_1,\dots,n_k}^2 \operatorname*{\sim}_{n\to\infty} \frac{1}{k^{2n}} \cdot k^{2n}\frac{k^{k/2}}{ (4\pi n)^{(k-1)/2}}
\]
where the equivalent is due to Richmond and Shallit~\cite{RS:08:numbertheory}.
\end{proof}
\noindent However, from~\cref{fact:23:2:holder} we want to get a hold on $\normtwo{{\p^\ast}^{-\max}_{-\eps}}$, not $\normtwo{\p^\ast}$ (since $\normone{{\p^\ast}^{-\max}_{-\eps}}^2 \geq 1-\Omega(\eps)$, we then will have our lower bound on $\norm{{\p^\ast}^{-\max}_{-\eps}}_{2/3}$). Fortunately, the two are related: namely, $\normtwo{{\p^\ast}^{-\max}_{-\eps}}\leq \normtwo{\p^\ast}$, so
$
    \frac{1}{\normtwo{{\p^\ast}^{-\max}_{-\eps}}} \geq \frac{1}{\normtwo{\p^\ast}}
$ which is the direction we need.

\noindent Combining the three facts above establishes (b), providing a lower bound of $q_{\rm hard}(n,k,\eps) = \bigOmega{ \frac{ (4\pi n)^{(k-1)/4}}{k^{k/4}\eps^2} }$ for testing identity to $\p^\ast$. It only remains to establish (c):

\begin{lemma}
There exists a (not necessarily efficient) agnostic learner for $\classpmd[n]{k}$, with sample complexity $q_{\rm agn}(n,k,\eps) = \frac{1}{\eps^2}\left( O(k^2\log n)+ \bigO{ \frac{k\log(k/\eps)}{\log\log(k/\eps)} }^k \right)$.
\end{lemma}
\begin{proof}
This is implied by a result of~\cite{DKS15c}, which establishes the existence of a (proper) $\eps$-cover $\mathcal{M}_{n,k,\eps}$ of $\classpmd[n]{k}$ such that
$
    \abs{ \mathcal{M}_{n,k,\eps} } \leq n^{O(k^2)}\cdot (1/\eps)^{\bigO{ \frac{k\log(k/\eps)}{\log\log(k/\eps)} }^{k-1}}
$. By standard arguments, this yields information-theoretically an agnostic learner with sample complexity $\bigO{\frac{\log\abs{ \mathcal{M}_{n,k,\eps} } }{\eps^2}}$.
\end{proof}

Having (a), (b), and (c), an application of~\cite[Theorem 18]{CDGR:16} yields that, as long as 
$
  q_{\rm agn}(n,k,\eps) = o( q_{\rm hard}(n,k,\eps) )
$ then testing membership to $\classpmd[n]{k}$ requires $\bigOmega{q_{\rm hard}(n,k,\eps)}$ samples as well. This in particular holds for $k = o(n^{c})$ (where e.g. $c<1/9$) and $\eps = 1/2^{O(n)}$.

\end{proof}

%% file: app-omitted.tex
\section{Omitted Proofs}\label{app:omitted}
In this appendix, we provide the proofs of the lemmas and technical results omitted in the main body.

\subsection{From~\cref{sec:prelim}}
\begin{proof}[Proof of~\cref{claim:ksiirv:l2:norm}]
By Plancherel, we have
$
  \normtwo{\p'}^2 = \frac{1}{M} \sum_{\xi=0}^{M-1} \dabs{\fourier{\p'}(\xi)}^2 = \frac{1}{M} \sum_{\xi=0}^{M-1} \dabs{\fourier{\p}(\xi)}^2
$, 
the second equality due to the definition of $\fourier{\p'}$. Indeed, for any $\xi\in[M]$, 
\begin{align*}
    \fourier{\p'}(\xi) &= \sum_{j=0}^{M-1} e^{-2i\pi \frac{j\xi}{M}} \p'(j) = \sum_{j=0}^{M-1} e^{-2i\pi \frac{j\xi}{M}} \sum_{\substack{j'\in\mathbb{N}\\j' = j \bmod M}} \p(j')
    = \sum_{j=0}^{M-1} \sum_{\substack{j'\in\mathbb{N}\\j' = j \bmod M}} e^{-2i\pi \frac{j'\xi}{M}} \p(j')
    \\
    &= \sum_{j\in\mathbb{N}} e^{-2i\pi \frac{j'\xi}{M}} \p(j')
    = \fourier{\p}(\xi)
\end{align*}
as $u\mapsto e^{-2i\pi u}$ is $1$-periodic. Since $\abs{\fourier{\p}(\xi)}\leq 1$ for every $\xi\in[M]$ (as $\fourier{\p}(\xi) = \E_{j\sim \p}[e^{-2i\pi \frac{j\xi}{M}}]$), we can upper bound the RHS as
\[
    \frac{1}{M} \sum_{\xi=0}^{M-1} \dabs{\fourier{\p}(\xi)}^2 \leq \frac{1}{M} \sum_{r\geq 0} \sum_{\xi : \frac{1}{2^{r+1}} < \abs{\fourier{\p}(\xi)} \leq \frac{1}{2^{r}}} \abs{\fourier{\p}(\xi)}^2
    \leq \frac{1}{M} \sum_{r\geq 0} \frac{1}{2^{2r}} \abs{\setOfSuchThat{ \xi\in[M] }{ \frac{1}{2^{r+1}} < \abs{\fourier{\p}(\xi)} }}\;.
\]
Invoking~\cref{lemma:FourierSupportLem}(ii) with parameter $\delta$ set to $\frac{1}{2^{r+1}}$, we get that $\abs{\setOfSuchThat{ \xi\in[M] }{ \frac{1}{2^{r+1}} < \abs{\fourier{\p}(\xi)} }} \leq 4Mk s^{-1} \sqrt{r+1}$, from which \[
    \normtwo{\p'}^2 \leq \frac{4k}{s} \sum_{r\geq 0} \frac{\sqrt{r+1}}{2^{2r}} \leq \frac{8k}{s}
\]
as desired.
\end{proof}

\subsection{From~\cref{sec:fourier:support:testing}}

\begin{proofof}{\cref{prop:l2:identity:tester}}
Letting $X_i$ denote the number of occurrences of the $i$-th domain element in the samples from $\p$, define $Z_i=(X_i-m\p^\ast(i))^2-X_i$. Since $X_i$ is distributed as $\Poi(m\cdot p_i),$  $\E[Z_i] = m^2(\p(i) - \p^\ast(i))^2$; thus, $Z$ is an unbiased estimator for $m^2\normtwo{ \p-\p^\ast }^2$. (Note that this holds even when $\p^\ast$ is allowed to take negative values.)

We compute the variance of $Z_i$ via a straightforward calculation involving standard expressions for the moments of a Poisson distribution: getting
\[
  \Var[Z]  = \sum_{i=1}^r \Var[Z_i] = \sum_{i=1}^r \left(4m^3(\p(i)-\p^\ast(i))^2 \p(i)  + 2m^2\p(i)^2\right).
\]
By Cauchy--Schwarz, and since $\sum_{i=1}^r \p(i)^2 \leq b$ by assumption, we have 
\begin{align*}
\sum_{i=1}^r (\p(i)-\p^\ast(i))^2 \p(i) &= \sum_{i=1}^r (\p(i)-\p^\ast(i))\cdot (\p(i)-\p^\ast(i)) \p(i)  \\
&\leq \sqrt{\sum_{i=1}^r (\p(i)-\p^\ast(i)) ^2 \sum_{i=1}^r \p(i)^2(\p(i)-\p^\ast(i))^2} \\
&\leq \sqrt{\sum_{i=1}^r (\p(i)-\p^\ast(i)) ^2 b \sum_{i=1}^r (\p(i)-\p^\ast(i)) ^2}
 = \sqrt{b} \normtwo{ \p-\p^\ast }^2
\end{align*}
and so 
 \[
 \Var[Z] \leq 4 m^3 \sqrt{b} \normtwo{ \p-\p^\ast }^2 + 2 m^2 b.
 \]
 
\noindent For convenience, let $\eta\eqdef\frac{1}{10}$, and write $\rho \eqdef \frac{\normtwo{ \p-\p^\ast }}{\eps}$ -- so that we need to distinguish $\rho \leq 1$ from $\rho \geq 2$. If $\rho \leq 1$, i.e. $\E[Z] \leq m^2\eps^2$, then 
 \[
    \Pr[ Z > (3-\eta)m^2\eps^2 ] = \Pr[ \lvert Z - \E[Z] \rvert > m^2\eps^2( ((3-\eta)-\gamma) - \rho^2 ) ]
 \]
 while if $\rho \geq 2$, i.e. $\E[Z] \geq 4m^2\eps^2$, then
 \[
    \Pr[ Z < (3+\eta)m^2\eps^2 ] = \Pr[ \E[Z] - Z > m^2( \lVert p-q\rVert_2^2 - (3+\eta)\eps^2 ) ] \leq \Pr[ \lvert Z - \E[Z] \rvert > m^2\eps^2( \rho^2 - (3+\eta) ) ].
 \]
In both cases, by Chebyshev's inequality, the test will be correct with probability at least~$3/4$ provided $m \geq c\sqrt{b}/\eps^2$ for some suitable choice of $c > 0$, since (where
\begin{align*}
    \Pr[ \lvert Z - \E[Z] \rvert > m^2\eps^2\lvert \rho^2 - (3\pm\eta) \rvert ] 
    &\leq \frac{\Var[Z]}{m^4\eps^4( \rho^2 - (3\pm\eta) )^2} \\
    &\leq \frac{ 4m^3 \sqrt{b} \rho^2\eps^2 + 2 m^2 b }{m^4\eps^4( \rho^2 - (3\pm\eta) )^2} 
    = \frac{\rho^2}{(\rho^2 - (3\pm\eta) )^2}\cdot\frac{ 4\sqrt{b}}{m\eps^2} + \frac{1}{(\rho^2 - (3\pm\eta) )^2}\cdot\frac{2 b }{m^2\eps^4} \\
    &\leq \frac{ 20\sqrt{b}}{m\eps^2} + \frac{5 b }{2m^2\eps^4}  \leq \frac{20}{c}+\frac{5}{2c^2} \leq \frac{1}{3}
\end{align*}
as $\max_{\rho\in[0,1]} \frac{\rho^2}{(\rho^2 - (3\pm\eta) )^2} \leq 5$ and $\max_{\rho\in[0,1]} \frac{1}{(\rho^2 - (3\pm\eta) )^2} \leq \frac{5}{4}$ and the last inequality holds for $c \geq 61$.
\end{proofof}

\subsection{From~\cref{sec:projections}}
\begin{proofof}{\cref{lemma:project:ksiirv}}
By Theorem 3.7 of \cite{DKS:16}, there is an algorithm that can compute an $\eps$-cover of all $(n,k)$-SIIRVs of size $n\left( k/\eps \right)^{O(k\log(1/\epsilon))}$ that runs in time $n\left( k/\eps \right)^{O(k\log(1/\epsilon))}$. Note the way the cover is given, allows us to compute the Fourier coefficients $\fourier{\q}(\xi)$ for any $\xi$ for each $\q \in \mathcal{C}$ in time $\poly(k/\eps)$.

 Since $\eps/\sqrt{|S|}=1/\poly(k/\eps)$, Step~\ref{step:cover:ksiirv} takes time $n\left( k/\eps \right)^{O(k\log(k/\epsilon))}$ and outputs a cover of size $n\left( k/\eps \right)^{O(k\log(k/\epsilon))}$. As each iteration takes time $|S|$, the whole algorithm takes $n\left( k/\eps \right)^{O(k\log(k/\epsilon))}$ time.

Note that each $\q$ that passes Step~\ref{step:cover:moment-test} is effectively supported on $I$ by (\ref{eq:bennett:application}) and has Fourier transform supported on $S$ by~\cref{claim:ksiirv:fourier:concentrated}.

\begin{itemize}
  \item Suppose that $\p\in\classksiirv[n]{k}$. Then there is a $(n,k)$-SIIRV $\q \in \mathcal{C}$ with $\dtv(\p,\q) \leq {\eps}/{5\sqrt{|S|}}$.
We need to show that if the algorithm considers $\q$, it accepts. From standard concentration bounds, one gets that the expectations of $\p$ and $\q$ are within $O(\eps\sqrt{\log(1/\eps)})$ standard deviations of $\p$ and the variances of $\p$ and $\q$ are within $O(\eps\log(1/\eps))$ multiplicative error.
Thus $\q$ passes the condition of Step \ref{step:cover:moment-test}. Since $\dtv(\p,\q) \leq \eps/(5\sqrt{|S|})$, we have that $|\fourier{\p}(\xi)-\fourier{\q}(\xi)| \leq \eps/(5\sqrt{|S|})$ for all $\xi$.
In particular, we have $\sum_{\xi \in S} |\fourier{\h}-\fourier{\q}|^2 \leq \eps^2/25$. Thus by the triangle inequality for $L_2$ norm, we have  $\sum_{\xi \in S} |\fourier{\h}-\fourier{\q}|^2 \leq (\eps/5 + 3\eps/25)^2 \leq (\eps/\sqrt{5})^2$. Thus the algorithm accepts.

  \item Now suppose that the algorithm accepts. We need to show that $\p$ has total variation distance at most $\eps$ from some $(n,k)$-SIIRV. We will show that $\dtv(\p,\q) \leq \eps$ for the $\q$ which causes the algorithm to accept.  Since the algorithm accepts, $\sum_{\xi \in S} |\fourier{\h}-\fourier{\q}|^2 \leq \eps^2/25$. 
For $x \notin S$, $\fourier{\h}(\xi)=0$ and so $\sum_{\xi \notin S} |\fourier{\h}-\fourier{\q}|^2 = \sum_{\xi \notin S} |\fourier{\q}|^2 \leq \eps^2/100$ by~\cref{claim:ksiirv:fourier:concentrated}. By Plancherel, the distributions $\q'\eqdef \q \bmod M$, $\h'\eqdef \h \bmod M$ satisfy 
\[
  \normtwo{ \q'-\h'}^2 = \frac{1}{M}\sum_{\xi=0}^{M-1} \dabs{ \fourier{\h}-\fourier{\q} }^2 \leq \frac{\eps^2}{20M}.
\] Thus $\dtv(\q',\h') \leq \frac{\eps}{4}$. By definition $\h$ has probability $0$ outside $I$ and by (\ref{eq:bennett:application}), $\q$ has at most $\frac{\eps}{5}$ probability outside $I$, Thus $\dtv(\q,\h) \leq \frac{\eps}{4}+\frac{\eps}{5} \leq \frac{\eps}{2}$ and by the triangle inequality $\dtv(\p,\q) \leq \dtv(\q,\h) + \dtv(\p,\h) \leq \eps/2 + 6\eps/25 \leq \eps$ as required.
\end{itemize}
\end{proofof}

 \begin{proofof}{\cref{lemma:project:pbd:smallvariance}}
We use Steps 4 and 5 of Algorithm \texttt{Proper-Learn-PBD} in \cite{DKS15b}. Step 5 checks if one of a system of polynomials has a solution. If such a solution is found, it corresponds to an $(n,2)$-SIIRV $\q$ that has $\sum_{|\xi| \leq \ell} |\fourier{\h}(\xi) - \fourier{\q}(\xi)|^2 \leq \eps^2/4$ and so we accept. If no systems have a solution, then there is no such $(n,2)$-SIIRV and so we reject. The conditions of this lemma are enough to satisfy the conditions of Theorem 11 of \cite{DKS15b}, though we need that the constant $C'$ used to define $|S|$ is sufficiently large to cover the $\ell=O(\log(1/\eps)$ from that paper. This theorem means that if $\p$ is a $(n,2)$-SIIRV, then we accept.

We need to show that if the algorithm finds a solution, then it is within $\eps$ of a Poisson Binomial distribution.  The system of equations ensures that $\sum_{|\xi| \leq \ell} |\fourier{\h}(\xi) - \fourier{\q}(\xi)|^2 \leq \eps^2/4$. Now the argument is similar to that for $(n,k)$-SIIRVs.
For $x \notin S$, $\fourier{\h}(\xi)=0$ and so $\sum_{\xi \notin S} |\fourier{\h}-\fourier{\q}|^2 = \sum_{\xi \notin S} |\fourier{\q}|^2 \leq \eps^2/100$ by~\cref{claim:ksiirv:fourier:concentrated}. By Plancherel, the distributions $\q'\eqdef \q \bmod M$, $\h'\eqdef \h \bmod M$ satisfy 
\[
  \normtwo{ \q'-\h'}^2 = \frac{1}{M}\sum_{\xi=0}^{M-1} \dabs{ \fourier{\h}-\fourier{\q} }^2 \leq \frac{\eps^2}{20M}.
\] Thus $\dtv(\q',\h') \leq \frac{\eps}{4}$. By definition $\h$ has probability $0$ outside $I$ and by (\ref{eq:bennett:application}), $\q$ has at most $\frac{\eps}{5}$ probability outside $I$, Thus $\dtv(\q,\h) \leq \frac{\eps}{4}+\frac{\eps}{5} \leq \frac{\eps}{2}$ and by the triangle inequality $\dtv(\p,\q) \leq \dtv(\q,\h) + \dtv(\p,\h) \leq \eps/2 + 6\eps/25 \leq \eps$ as required.
 \end{proofof}

%% file: app-logconcavelearning.tex
\section{Learning Discrete Log-Concave Distributions in Hellinger Distance}\label{appendix:log:concave}

Recall that the Hellinger distance between two probability distributions over a domain $\mathbb{D}$ is defined as 
\[
\hellinger{p}{q} \eqdef \frac{1}{\sqrt{2}}\normtwo{\sqrt{p}-\sqrt{q}}
\]
where the 2- norm is to be interpreted as either the $\lp[2]$ distance or $L^2$ distance between the pmf or pdf's of $p,q$, depending on whether $\mathbb{D}$ is $\Z$ or $\R$. In particular, one can extend this metric to the set of \emph{pseudo}-distributions over $\mathbb{D}$, relaxing the requirement that the measures sum to one. We let $\mathcal{F}_{\mathbb{D}}$ denote the set of pseudo-distributions over $\mathbb{D}$. The \emph{bracketing entropy} of a family of functions $\mathcal{G}\subseteq\R^{\mathbb{D}}$ with respect to the Hellinger distance (for parameter $\eps$) if then the minimum cardinality of a collection $\class$ of pairs $(g_L,g_U)\in \mathcal{F}_{\mathbb{D}}^2$ such that every $f\in\mathcal{G}$ is ``bracketed'' between the elements of some pair in $\class$:
\[
    \bracketing{\eps}{\mathcal{G}} \eqdef \min\setOfSuchThat{ N\in\N }{ \exists \class \subseteq \mathcal{F}_{\mathbb{D}}^2,\ \abs{\class}=N,\ \forall f\in\mathcal{G}, \exists (g_L,g_U)\in\class\text{ s.t. } g_L\leq f\leq g_U \text{ and } \hellinger{g_L}{g_H}\leq \eps }
\]

\begin{theorem}\label{theo:mle:logconcave}
Let $\hat{p}_m$ denote the maximum likelihood estimator (MLE) for discrete log-concave distributions on a sample of size $m$. Then, the minimax supremum risk satisfies
\[
    \sup_{p\in\classlogconcave_n}\shortexpect_p[ \hellinger{\hat{p}_m}{p}^2 ] = \bigO{ \frac{\log^{4/5} (mn)}{m^{4/5}} }.
\]
\end{theorem}

Note that it is known that for \emph{continuous} log-concave distributions over $\R$, the rate of the MLE is $O(m^{-4/5})$~\cite{KS:16}; this result, however, does not generalize to discrete log-concavity, as it crucially relies on a scaling argument which does not work in the discrete case. On the other hand, one can derive a rate of convergence to learn discrete log-concave distributions in \emph{total variation distance} (using another estimator than the MLE), getting again $O(m^{-4/5})$ in that case~\cite{DKS:16:LLCV}. However, due to the loose upper bound relating total variation and Hellinger distance, this latter result only implies an $O(m^{-2/5})$ convergence rate in Hellinger distance, which is quadratically worse than what we would hope for.

Thus, the result above, while involving a logarithmic dependence on the support size, has the advantage of getting the ``right'' rate of convergence. (While this additional dependence does not matter for our purposes, we believe a modification of our techniques would allow one to get rid of it, obtaining a rate of $\tildeO{m^{-4/5}}$ instead.) \new{We however conjecture that the tight rate of convergence should be $O(m^{-4/5})$, as in the continuous case (i.e., without the dependence on the domain size $n$ nor the extra logarithmic factors in $m$).}\medskip

In order to prove~\cref{theo:mle:logconcave}, we obtain along the way several interesting results on discrete (and continuous) log-concave distributions, namely a bound on their bracketing entropy (\cref{theo:bracketing:hellinger}) and an approximation result (\cref{theo:approx:hellinger}), which we believe are of independent interest.\medskip

In what follows, $\mathbb{D}$ will denote either $\R$ or $\Z$; we let $\classlogconcave(\mathbb{D})$ denote the set of log-concave distributions over $\mathbb{D}$, and $\classlogconcave_n\subseteq \classlogconcave(\Z)$ be the subset of log-concave distributions supported on $\{0,\dots,n-1\}$.
\begin{theorem}\label{theo:bracketing:hellinger}
  For every $\eps\in(0,1)$,
  \[
      \bracketing{\eps}{\classlogconcave_n} \leq \left(\frac{n}{\eps}\right)^{O(1/\sqrt{\eps})}
  \]  
\end{theorem}

A crucial element in to establish~\cref{theo:bracketing:hellinger} will be the following theorem, which shows that log-concave distributions are well-approximated (in Hellinger distance) by piecewise-constant pseudo-distributions with few pieces:
 \begin{theorem}\label{theo:approx:hellinger}
  Let $\mathbb{D}$ be either $\R$ or $\Z$. For every $p\in\classlogconcave(\mathbb{D})$ and $\eps\in(0,1)$, there exists a pseudo-distribution $g$ such that (i) $g$ is piecewise-linear with $\bigO{1/\sqrt{\eps}}$ pieces; (ii) $g$ is supported on an interval $[a,b]$ with $p(\mathbb{D}\setminus[a,b]) = O(\eps^2)$; and (iii) $\hellinger{p}{g}\leq \eps$. (Moreover, one can choose to enforce $g\leq p$, or $p\leq g$, on $[a,b]$).
\end{theorem}

The proof of~\cref{theo:approx:hellinger} will be very similar to that of~\cite[Theorem 12]{DKS:16:LLCV}; specifically, we will use the following (reformulation of a) lemma due to Diakonikolas, Kane, and Stewart:
\begin{lemma}[{\cite[Lemma 14]{DKS:16:LLCV}, rephrased}]\label{lemma:logconcave:pointwise:approx}
  Let $\mathbb{D}$ be either $\R$ or $\Z$. Let $f$ be a log-concave function defined on an interval $I\subseteq \mathbb{D}$, and suppose that $f(I)\subseteq [a,2a]$ for some constant $a>0$. Furthermore, suppose that the logarithmic derivative of $f$ (or, if $\mathbb{D}=\Z$, the log-finite difference of $f$) varies by at most $1/\abs{I}$ on $I$; then, for any $\eps\in(0,1)$ there exists two piecewise linear functions $g^\ell,g^u\colon I\mapsto \R$ with $\bigO{1/\sqrt{\eps}}$ pieces such that
  \begin{equation}\label{eq:logconcave:pointwise:approx}
      \abs{f(x)-g^j(x)} = \bigO{\eps}f(x), \qquad j\in\{\ell,u\}
  \end{equation}
  for all $x\in I$, and with $g^\ell\leq f\leq g^u$.
\end{lemma}
\begin{proof}
  Observe that it suffices to establish~\cref{eq:logconcave:pointwise:approx} for a piecewise linear function $g\colon I\mapsto \R$ with $\bigO{1/\sqrt{\eps}}$ pieces; indeed, then in order to obtain $g^\ell,g^u$ from $g$, it will be sufficient to scale it by respectively $(1+\alpha\eps)^{-1}$ and  $(1+\alpha\eps)$ (for a suitably big absolute constant $\alpha>0$), thus ensuring both~\cref{eq:logconcave:pointwise:approx} and $g^\ell\leq f\leq g^u$. We therefore focus hereafter on obtaining such a pseudo-distribution $g$.

  For ease of notation, we write $h$ for the logarithmic derivative (or log-finite difference) of $f$ (e.g., in the continuous case, $h=(\ln f)'$). By rescaling $f$, we may assume without loss of generality that $a=1$. Note that $h$ is then bounded on $I$, i.e. $\abs{h}\leq {c}/{\abs{I}}$ for some absolute constant $c>0$.  We now partition $I$ into subintervals $J_1,J_2,\ldots,J_\ell$ so that (i) each $J_i$ has length at most $\eps^{1/2}\abs{I}$, and (ii) $h$ varies by at most $\eps^{1/2}/\abs{I}$ on each $J_i$. This can be achieved with $\ell=\bigO{1/\sqrt{\eps}}$ by placing an interval boundary every $\eps^{1/2}\abs{I}$ distance as well as every time $h$ passes a multiple of $\eps^{1/2}/\abs{I}$.

We now claim that on each interval $J_i$ there exists a linear function $g_i$ so that $\abs{g_i(x)-f(x)} = O(\eps)f(x)$ for all $x\in J_i$. Letting $g$ be $g_i$ on $J_i$ will complete the proof. Fix any $i$, and write $J_i=[s_i,t_i]$. Letting $\alpha_0\in h(J_i)$ be an arbitrary value in the range spanned by $h$ on $J_i$, observe that for any $x\in J_i$ there exists $\alpha_x\in h(J_i)$ such that 
\[
  f(x) = f(s_i) e^{\alpha_x(x-s_i)}
\]
from which we have
\begin{align*}
  f(x) &= f(s_i) e^{\alpha_0(x-s_i)+(\alpha_x-\alpha_0)(x-s_i)}
  = f(s_i) e^{\alpha_0(x-s_i)}e^{(\alpha_x-\alpha_0)(x-s_i)}\\
  &= f(s_i) \left(1+\alpha_0(x-s_i)+O(\eps)\right)(1+O((\alpha_x-\alpha_0)(x-s_i)))\\
  &= f(s_i) \left(1+\alpha_0(x-s_i)+O(\eps)\right)(1+O(\eps))\\
  &= f(s_i)+\alpha_0f(s_i)(x-s_i)+O(\eps)
\end{align*}
recalling that $\abs{\alpha_0},\abs{\alpha_x}=O(1/\abs{I})$, $\abs{x-s_i}\leq \eps^{1/2}\abs{I}$, and $\abs{\alpha_x-\alpha_0}\leq \eps^{1/2}/\abs{I}$, so that $\abs{\alpha_0(x-s_i)} = O(\eps^{1/2})$ and $\abs{(\alpha_x-\alpha_0)(x-s_i)} = O(\eps)$. This motivates defining the affine function $g_i$ as
\[
    g_i(x) \eqdef f(s_i)+\alpha_0f(s_i)(x-s_i), \qquad x\in J_i
\]
from which
\begin{align*}
    \abs{\frac{f(x)-g_i(x)}{f(x)}} 
    &= \abs{1-\frac{f(s_i)+\alpha_0f(s_i)(x-s_i)}{f(s_i) e^{\alpha_x(x-s_i)}}}
    = \abs{1-\frac{1+\alpha_0(x-s_i)}{e^{\alpha_x(x-s_i)}}} \\
    &= \abs{1-\frac{1+\alpha_0(x-s_i)}{1+\alpha_x(x-s_i)+O(\eps)}} 
    = \abs{1-\left( 1+\alpha_0(x-s_i)\right)\left( 1-\alpha_x(x-s_i)+O(\eps)\right)} \\
    &= \abs{(\alpha_x-\alpha_0)(x-s_i)+O(\eps)} = O(\eps)
\end{align*}
as claimed. This concludes the proof.
\end{proof}

\noindent We will also rely on the following proposition, from the same paper:
\begin{proposition}[{\cite[Proposition 15]{DKS:16:LLCV}}]\label{proposition:logconcave:interval:partition}
Let $f$ be a log-concave distribution on $\mathbb{D}$ (as before, either $\R$ or $\Z$). Then there exists a partition of $\mathbb{D}$ into disjoint intervals $I_1, I_2,\ldots$ and a constant $C>0$ such that
\begin{itemize}
\item $f$ satisfies the hypotheses of~\cref{lemma:logconcave:pointwise:approx} on each $I_i$.
\item For each $m$, there are most $Cm$ values of $i$ so that $f(I_i) > 2^{-m}$.
\end{itemize}
(Moreover, $f$ is monotone on each $I_i$.)
\end{proposition}

\noindent We are now ready to prove~\cref{theo:approx:hellinger}:
\begin{proofof}{\cref{theo:approx:hellinger}}
Fix any $\eps\in(0,1)$, and $p\in\classlogconcave(\mathbb{D})$. We divide $\mathbb{D}$ into intervals as described in~\cref{proposition:logconcave:interval:partition}. Call these intervals $I_1,I_2,\ldots$ sorted so that $p(I_i)$ is decreasing in $i$. Therefore, we have that $p(I_m) \leq 2^{-m/C}$.

For $1 \leq m \leq M\eqdef 2C\log(1/\eps)$, let $\eps_m\eqdef \eps 2^{m/(3C)}$; we use~\cref{lemma:logconcave:pointwise:approx} to approximate $p$ in $I_m$ by two piecewise linear functions $g^{\ell}_m, g^u_m$ so that (i) $g^j_m$ has at most $O(1/\sqrt{\eps_m})$ pieces and (ii) $p$ and $g^j_m$ are, on $I_m$, within a multiplicative $(1\pm O(\eps_m))$ factor with $g^{\ell}_m \leq p\leq g^u_m$. For $j\in\{\ell,u\}$, let $g^j$ be the piecewise linear function that is $g^j_m$ on $I_m$ for $1\leq m\leq M$, and $0$ elsewhere. $g^j$ is then piecewise linear on
\[
\sum_{m=1}^{M} O(\eps_m^{-1/2}) = \sum_{m=1}^{M} \bigO{ \eps^{-1/2} 2^{-\frac{m}{6C}} } = O(\eps^{-1/2})
\]
intervals.

Let $I$ be defined as the smallest interval such that $\bigcup_{m=1}^M I_m\subseteq I$. By definition, $g$ is $0$ outside of $I$, and moreover the total mass of $p$ there is
\[
  \sum_{m=M+1}^\infty p(I_m) \leq \sum_{m=M+1}^\infty \frac{1}{2^{m/C}} = \bigO{2^{-M/C}} = \bigO{\eps^2}
\]
By replacing $g^j$ by $\max(g^j,0)$, we may ensure that it is non-negative (while at most doubling the number of pieces without increasing the distance from $p$). This establishes the first two items of the theorem; we now turn to the third.

The Hellinger distance between $p$ and $g^j$ satisfies, letting $J\eqdef \bigcup_{m=1}^M I_m$,
\begin{align*}
  2\hellinger{p}{g^j}^2 &= \normtwo{\sqrt{p}-\sqrt{g^j}}^2
    = \int_{\mathbb{D}}\left(\sqrt{p(x)}-\sqrt{g^j(x)}\right)^2 \mu(dx) \\
    &= \int_{\mathbb{D}\setminus J}\left(\sqrt{p(x)}-\sqrt{g^j(x)}\right)^2 \mu(dx)+\int_{J}\left(\sqrt{p(x)}-\sqrt{g^j(x)}\right)^2 \mu(dx)\\
    &= \int_{\mathbb{D}\setminus J}p(x) \mu(dx)+\sum_{m=1}^M\int_{I_m}p(x)\left(1-\sqrt{1\pm O(\eps_m)}\right)^2 \mu(dx)\\
    &\leq O(\eps^2) + \sum_{m=1}^M\int_{I_m}p(x)\left(1-\sqrt{1\pm O(\eps_m)}\right)^2 \mu(dx) \\
    &= O(\eps^2) + \sum_{m=1}^M\int_{I_m}p(x)O(\eps_m^2) \mu(dx) 
    = O(\eps^2) + \sum_{m=1}^M \bigO{\eps_m^2 p(I_m)} \\
    &= O(\eps^2) + \sum_{m=1}^M \bigO{\eps^2 2^{\frac{2m}{3C}}2^{\frac{-m}{C}}} 
    = O(\eps^2) + \sum_{m=1}^M \bigO{\eps^2 2^{\frac{-m}{3C}}} \\
    &= O(\eps^2) + O(\eps^2) = O(\eps^2)
\end{align*}
establishing the third item. (By dividing $\eps$ by a sufficiently big absolute constant before applying the above, one gets (i), (ii), and (iii) with $\hellinger{p}{g^j} \leq \eps$ as desired.) \new{For technical reasons (that we will need in the proof of~\cref{theo:bracketing:hellinger}), instead of defining $[a,b]$ to be our interval $I$, we choose $[a,b]$ to be $I$ augmented with up to two of the remaining $I_m$'s (those directly on the left and right of $I$, defining $g^\ell_m, g^u_m$ on these two additional pieces as before by~\cref{lemma:logconcave:pointwise:approx}). This does not change the fact that the piecewise linear function obtained on $[a,b]$ has $O(\eps^{-1/2})$ pieces (we only added $o(\eps^{-1/2})$ pieces), and $p(\mathbb{D}\setminus [a,b])\leq p(\mathbb{D}\setminus I) = O(\eps^2)$. Finally, it is easy to see that this only changes, as per the computation above, the Hellinger distance by $O(\eps^2)$ as well. (The advantage of this technicality is that now, the two end intervals in the union constituting  $[a,b]$ have each total probability mass $O(\eps^2)$ under $p$, which will come in handy later.)} It then only remains to choose $g$ to be either $g^\ell$ or $g^u$, depending on whether one wants a lower- or upperbound on $f$ (on $[a,b]$).
\end{proofof}

\noindent We can finally prove~\cref{theo:bracketing:hellinger}:
\begin{proofof}{\cref{theo:bracketing:hellinger}}
We can slightly strengthen the proof of~\cref{theo:approx:hellinger} for the case of $\classlogconcave_n$, by imposing some restriction on the form of the `approximating distributions'' $g$. Namely, for any $\eps\in(0,1)$, fix any $p\in\classlogconcave_n$ and consider the construction of $g^\ell,g^u$ as in the proof of~\cref{theo:approx:hellinger}. Clearly, we can assume $[a,b]\subseteq\{0,\dots,n-1\}$. 

Now, we modify $g^j$ as follows (for $j\in\{\ell,u\}$): for $1\leq m\leq M$, consider the interval $I_m=[a_m,b_m]$, and the corresponding ``piece'' $g^j_m$ of $g$ on $I_m$. We let $\tilde{g}^j_m$ be the pseudo-distribution defined from $g^j_m$ as follows: it is affine on $I_m$, with
\[
\tilde{g}^u_m(a_m) \eqdef \clg{g^u(a_m)\frac{M\abs{I_m}}{2\eps^2}}\frac{2\eps^2}{M\abs{I_m}}, \qquad \tilde{g}^u_m(a_m) \eqdef \clg{g^u(b_m)\frac{M\abs{I_m}}{2\eps^2}}\frac{2\eps^2}{M\abs{I_m}}
\]
and
\[
\tilde{g}^\ell_m(a_m) \eqdef \flr{g^\ell(a_m)\frac{M\abs{I_m}}{2\eps^2}}\frac{2\eps^2}{M\abs{I_m}}, \qquad \tilde{g}^\ell_m(a_m) \eqdef \flr{g^\ell(b_m)\frac{M\abs{I_m}}{2\eps^2}}\frac{2\eps^2}{M\abs{I_m}}
\]
i.e. $g^j_m$ is $g$ ``rounded up'' (resp. down) to the near multiple of $\frac{\eps^2}{M\abs{I_m}}$ on the endpoints. We then let $\tilde{g}^j$ be the correspond piecewise-affine pseudo-distribution defined by piecing together the $\tilde{g}^j_m$'s. Clearly, by construction $\tilde{g}^\ell$ and $\tilde{g}^u$ still satisfies (i) and (ii) of~\cref{theo:approx:hellinger}, and $\tilde{g}^\ell\leq p\leq\tilde{g}^u$. As for (iii), observe that at all $1\leq m\leq M$ and $k\in I_m$ we have $\abs{\tilde{g}^j(k)-g^j(k)} \leq \frac{2\eps^2}{M\abs{I_m}}$, from which
\[
  \hellinger{p}{\tilde{g}^j} \leq \hellinger{p}{g^j} + \hellinger{g}{\tilde{g}^j}
  \leq \eps + \sqrt{\totalvardist{g^j}{\tilde{g}^j}}
  \leq \eps + \sqrt{\frac{1}{2}\sum_{m=1}^M \abs{I_m}\cdot\frac{2\eps^2}{M\abs{I_m}}}
  = 2\eps
\]
showing that we get (up to a constant factor loss in the distance) (iii) as well. Given this, we get that specifying $(\tilde{g}^\ell,\tilde{g}^u)$ can be done by the list of the $O(1/\sqrt{\eps})$ endpoints along with the value of each $\tilde{g}^j$ for all of these endpoints. Now, given the two endpoints, one gets the size of the corresponding interval $I_m$ (which is at most $n$), and the two values to specify are a multiple of $\eps^2/(M\abs{I_m})$ in $[0,1]$. (If we were to stop here, we would get the existence of an $\eps$-cover $\class'_\eps$ of $\classlogconcave_n$ in Hellinger distance of size $(n/\eps)^{O(1/\sqrt{\eps})}$.)

\paragraph{One Last Step: Outside $[a,b]$.} To get the bracketing bound we seek, we need to do one last modification to our pair $(\tilde{g}^\ell,\tilde{g}^u)$. Specifically, in the above we have one issue when approximating $p$: namely, that outside of their common support $\{a,\dots,b\}$, both $\tilde{g}^j$'s are $0$. While this is fine for the lower bound $\tilde{g}^\ell$, this is not for $\tilde{g}^u$, as it needs to dominate $p$ outside of $\{a,\dots,b\}$ as well, where $p$ may have $O(\eps^2)$  probability mass. Thus, we need to adapt the construction above, as follows (we treat the setting of $\tilde{g}^u$ on $\{b+1,\dots,n\}$, the definition on $\{0,\dots, a-1\}$ is similar).

First, observe if $p(b+1)=0$, we are done, as then by monotonicity we must have $(k)=0$ for all $k\geq b+1$, and so setting $\tilde{g}^u=0$ on $\{b+1,\dots,n\}$ suffices. Thus, we hereafter assume $p(b+1)>0$; and, for $b+1\leq k\leq n$, set
\[
    \tilde{g}^u(k) \eqdef \alpha e^{\beta(k-(b+1))}
\]
where $\alpha \eqdef \clg{p(b+1)\frac{n}{2\eps^2}}\frac{2\eps^2}{n}$ and $\beta \eqdef \clg{\frac{n}{\eps}\ln \frac{p(b+2)}{p(b+1)}}\frac{\eps}{n}$ (so that $\beta \leq 0$). Then $\tilde{g}^u(b+1)\geq p(b+1)$, and for $b+1<k\leq n$
\[
    \frac{\tilde{g}^u(k)}{\tilde{g}^u(k-1)} = e^\beta \geq \frac{p(b+2)}{p(b+1)} \geq \frac{p(k)}{p(k-1)}
\]
(the last inequality due to the log-concavity of $p$). This implies $\tilde{g}^u\geq p$ on $\{b+1,\dots,n\}$ as desired; and, thanks to the rounding, there are only $O(n/\eps^2)$ different possibilities for the tail of $\tilde{g}^u$. 
In view of bounding the Hellinger distance between $p$ and $\tilde{g}^u$ added by this modification, which is upper bounded by the (square root) of the total variation distance this added, recall that $p(\{b+1,\dots,n\})=O(\eps^2)$ by construction, and that
\[
    \tilde{g}^u(\{b+1,\dots,n\}) = \sum_{k=b+1}^n \alpha e^{\beta(k-(b+1))} = \frac{\alpha}{1-e^\beta}.
\]
Thus, the Hellinger distance incurred on $\{b+1,\dots,n\}$ is at most $\sqrt{O(\eps^2)+\frac{\alpha}{1-e^\beta}}$; and to conclude, it only remains to show that $\frac{\alpha}{1-e^\beta} = O(\eps^2)$.

To show this last point, let $I_m=[c,b]$ be the rightmost interval in the decomposition from~\cref{proposition:logconcave:interval:partition}. Recall that we are guaranteed that $p$ is non-increasing on $I_m$; further, by inspection of the proof of~\cite[Proposition 15]{DKS:16:LLCV}, we also have that $I_m$ is \emph{maximal}, in the sense that $b$ is the rightmost point $k$ such that $[c,k]$ satisfies the assumptions of~\cref{lemma:logconcave:pointwise:approx}. Using first the monotoncity, we have
\[
    p(b+1)\leq p(b)\leq \frac{p(I_m)}{b-c} \leq \frac{O(\eps^2)}{b-c}
\]
that last inequality by construction (from the technicality we enforced in the end of the proof of~\cref{theo:approx:hellinger}); and therefore $\alpha \leq \frac{O(\eps^2)}{b-c} + \frac{\eps^2}{n} = \frac{O(\eps^2)}{b-c}$.

\noindent In order to obtain an upper bound on $\beta$, we rely on the maximality of $I_m$, leading to two cases to consider:
\begin{itemize}
  \item The first is that $p(b+1) < \frac{p(c)}{2}$; in which case $p(b+2) \leq p(b+1) < \frac{p(c)}{2}$; which implies that
\[
    \frac{1}{2} > \frac{p(b+2)}{p(c)} = \frac{p(b+2)}{p(b+1)}\cdot\frac{p(b+1)}{p(b)}\cdots\frac{p(c+1)}{p(c)} \geq \left(\frac{p(b+2)}{p(b+1)}\right)^{b-c+2}
\]
the last inequality by log-concavity. In turn, we get
\[
    \beta \leq \ln\frac{p(b+2)}{p(b+1)} + \frac{\eps}{n} \leq -\frac{\ln 2}{b-c+2}+ \frac{\eps}{n}.
\]
  \item The second is that $\ln\frac{p(c+1)}{p(c)} - \ln\frac{p(b+1)}{p(b)} > \frac{1}{b-c+1}$. In this case,
  \[
      \ln \frac{p(b+2)}{p(b+1)} \leq \ln \frac{p(b+1)}{p(b)} < \ln\frac{p(c+1)}{p(c)} - \frac{1}{b-c+1} \leq - \frac{1}{b-c+1} < -\frac{\ln 2}{b-c+2}
  \]
  (the last inequality as $b-c \geq 0$) and therefore $\beta\leq -\frac{\ln 2}{b-c+2}+ \frac{\eps}{n}$ as in the first case.
\end{itemize}
Combining these two bounds, we obtain
\[
    \frac{\alpha}{1-e^\beta} \leq \frac{O(\eps^2)}{b-c} \cdot \frac{1}{1-e^{\frac{\eps}{n}}e^{-\frac{\ln 2}{b-c+2}}} = O(\eps^2)
\]
the last inequality for $\eps < \frac{\ln 2}{2}$ (using the fact that $1\leq b-c\leq n$). This concludes the proof: as discussed, we then have that our setting of $\bar{g}^u$ outside of $[a,b]$ only causes an addition Hellinger distance of $\sqrt{O(\eps^2)+\frac{\alpha}{1-e^\beta}} = \sqrt{O(\eps^2)}=O(\eps)$.

\end{proofof}

We are, at last, ready to prove our main theorem:
\begin{proofof}{\cref{theo:mle:logconcave}}
    Recall the following theorem, due to Wong and Shen~\cite{WS:95} (see also~\cite[Theorem 7.4]{vdG:00},~\cite[Theorem 17]{KS:16}):
    \begin{theorem}[{\cite[Theorem 2]{WS:95}}]
        There exist positive constants $\tau_1,\tau_2,\tau_3,\tau_4>0$ such that, for all $\eps\in(0,1)$, if
        \begin{equation}\label{eq:mle:logconcave:condition}
            \int_{\eps^2/2^8}^{\sqrt{2}\eps} \sqrt{\bracketing{u/\tau_1}{\mathcal{G}}}\, du \leq \tau_2 m^{1/2}\eps^2
        \end{equation}
        and $\tilde{p}_n$ is an estimator that approximates $\hat{p}_m$ within error $\eta$ (i.e., solves the maximization problem within additive error $\eta$) with $\eta \leq \tau_3\eps^2$, then
        \[
            \probaOf{ \hellinger{\tilde{p}_m}{p} \geq \eps } \leq 5\exp(-\tau_4 m \eps^2).
        \]
    \end{theorem}
    To apply this theorem, define the function $J_n\colon(0,1)\to\R$ by $J(x) \eqdef \int_{x^2}^x \sqrt{\ln\frac{n}{u}}u^{-1/4}\, du$. By (tedious) computations, one can verify that $J_n(x) \sim_{x\to 0} \frac{4}{3}x^{3/4}\sqrt{\ln\frac{n}{x}}$; this, combined with the bound of~\cref{theo:bracketing:hellinger}, yields that for any $\eps\in(0,1)$
    \[
        \int_{\eps^2/2^8}^{\sqrt{2}\eps} \sqrt{\bracketing{u/\tau_1}{\classlogconcave_n}}\, du = \bigO{\eps^{3/4}\sqrt{\ln\frac{n}{\eps}}}.
    \]
    Thus, setting, for $m\geq 1$, $\eps_m \eqdef C m^{-2/5}(\ln(mn))^{2/5}$ for a sufficiently big absolute constant $C>0$ ensures that $\eps_m$ satisfies~\eqref{eq:mle:logconcave:condition}. Let $\rho_{m} \eqdef 1/\eps_m$. It follows that any estimator which, on a sample of size $m$, approximates the log-concave MLE to within an additive $\eta_m\eqdef \tau_3\eps_m^2$ has minimax error
    \begin{align*}
        \rho^2_{m}\sup_{p\in\classlogconcave_n}\shortexpect_p[ \hellinger{\tilde{p}_m}{p}^2 ] &
        = \sup_{p\in\classlogconcave_n} \int_0^\infty {  \probaOf{ \rho^2_{m} \hellinger{\tilde{p}_n}{p}^2 \geq t } }\, dt \\
        &= \sup_{p\in\classlogconcave_n} \int_0^\infty {  \probaOf{  \hellinger{\tilde{p}_n}{p} \geq \sqrt{t}\rho^{-1}_{m} } }\, dt \\
        &\leq 1+ \sup_{p\in\classlogconcave_n} \int_1^\infty {  \probaOf{  \hellinger{\tilde{p}_n}{p} \geq \sqrt{t}\rho^{-1}_{m} } }\, dt \\
        &= 1+ \sup_{p\in\classlogconcave_n} \int_1^\infty {  \probaOf{  \hellinger{\tilde{p}_n}{p} \geq \sqrt{t}\eps_{m} } }\, dt \\
        &\leq 1+ 5\sup_{p\in\classlogconcave_n} \int_1^\infty \exp(-\tau_4 m t\eps_m^2)\, dt \\
        &= 1+ 5\sup_{p\in\classlogconcave_n} \int_1^\infty \exp(-\tau_4 Cm^{1/2}\ln(mn) t)\, dt \\
        &= O(1)
    \end{align*}
    where we used the fact that if $\eps_t > \eps_m$, then $\eps_t$ satisfies~\eqref{eq:mle:logconcave:condition} as well (and applied it to $\eps_t = \sqrt{t}\eps_{m}$). This concludes the proof.
\end{proofof}